\documentclass[12pt]{article}
\usepackage[letter paper, margin=1 in]{geometry}
\usepackage[utf8]{inputenc}
\usepackage{graphicx}
\usepackage{amsmath}
\usepackage[version=4]{mhchem}
\usepackage{siunitx}
\usepackage{longtable,tabularx}
\usepackage{microtype}
\usepackage{subcaption}
\usepackage{booktabs} 
\usepackage{float}

\usepackage{amsmath,amssymb,amsthm,mathtools}
\usepackage{bm,bbm}
\usepackage{algorithm,algpseudocode}
\usepackage{hyperref}
\usepackage{cleveref}
\usepackage[numbers,sort&compress]{natbib}

\usepackage{url}
\newcommand{\mbb}{\mathbb}
\newcommand{\mbf}{\mathbf}
\newcommand{\mcl}{\mathcal}

\newcommand{\R}{\mathbb{R}}

\newcommand{\T}{\textnormal}
\newcommand{\x}{\mathbf{x}}
\newcommand{\y}{\mathbf{y}}
\newcommand{\X}{\mathcal{X}}
\newcommand{\D}{\mathcal{D}}

\usepackage[normalem]{ulem} 

\newtheorem{proposition}{Proposition}
\newtheorem{assumption}{Assumption}
\newtheorem{theorem}{Theorem}

\newtheorem{lemma}{Lemma}

\title{Surrogate-Guided Adaptive Importance Sampling for Failure Probability Estimation}
\author{%
  Ashwin Renganathan\thanks{Aerospace Engineering and the Institute of Computational and Data Science (ICDS), Penn State, University Park, PA. Corresponding author.}%
  \and
  Annie S. Booth\thanks{Department of Statistics, Virginia Tech, Blacksburg, VA.}%
}
\date{}

\begin{document}

\maketitle

\begin{abstract}
We consider the sample efficient estimation of failure probabilities from expensive oracle evaluations of a limit state function via importance sampling (IS). In contrast to conventional ``two-stage'' approaches, which first train a surrogate model for the limit state and then construct an IS proposal to estimate the failure probability using separate oracle evaluations, we propose a ``single-stage'' approach where a Gaussian process surrogate and a surrogate for the optimal (zero-variance) IS density are trained from shared evaluations of the oracle, making better use of a limited budget. With this approach, small failure probabilities can be learned from relatively few oracle evaluations.  We propose \emph{kernel density estimation adaptive importance sampling} (\texttt{KDE-AIS}), which combines Gaussian process surrogates with kernel density estimation to adaptively construct the IS proposal density, leading to sample efficient estimation of failure probabilities. We show that the \texttt{KDE-AIS} density asymptotically converges to the optimal zero-variance IS density in total variation. Empirically, \texttt{KDE-AIS} enables accurate and sample efficient estimation of failure probabilities, outperforming state-of-the-art competitors including previous work on Gaussian process based adaptive importance sampling.
\end{abstract}


\section{Introduction}
\label{sec:intro}

The problem of estimating the probability of a rare event using data queried from an expensive blackbox computer model (``oracle'') simulating the event finds ubiquitous applications in climate science~\cite{naveau2020statistical}, engineering reliability analysis~\cite{renganathan2023camera,booth2024actively}, and geophysics~\cite{love2012credible}, to name a few.
Let $g(\x) :\X \rightarrow \R$ be an expensive oracle, with inputs $\x \in \X \subset \R^d$; we assume $\X$ is a compact set and $g$ is bounded above and below in $\X$. In the present context, $g$ is called a ``limit state'' function, with a threshold $t \in \R$, with $F = \{\x : g(\x) > t,~\x \in \mcl{X}\}$ being an event of interest (typically system failure). Let $(\Omega,\mcl{F},\mbb{P})$ be a probability space, and let
$X:(\Omega,\mcl{F})\to(\R^{d},\mcl{B}(\R^{d}))$ be an $\R^{d}$-valued random vector.
Assume $X$ admits a known density $p:\R^{d}\to[0,\infty)$ with respect to the Lebesgue measure. Then, we are interested in estimating the ``failure probability''
\begin{equation}
P_F = \mbb{P}(X\in F)=\int_{F} p(\x)\,\mathrm{d}\x,~ \forall ~F\in\mcl{B}(\R^{d}) \equiv \int_{\mcl{X}} \mathbbm{1}_{\{g(\x) > t\}} p(\x) \; d\x.
\label{eqn:fp_def}
\end{equation}
We specifically consider a situation where $g(\x)>t$ falls in the tails of $p$, and hence is a \emph{rare} event according to $p$. The overarching goal of this work is to estimate $P_F$ accurately with as few oracle evaluations as possible (100's as opposed to 1000's as is typical in the literature).

If $g$ is an oracle, then $P_F$ is not known in closed form and may be estimated with a naive Monte Carlo (MC) approximation:
\[
    P_F \approx \widehat P_F^\T{MC} = \frac{1}{N}\sum_{i=1}^N \mathbbm{1}_{\{g(X_i) > t\}}
    \quad\T{for}\quad X_i\sim p(\x), \;\; i=1,\dots,N.
\]
For rare event probability estimation, the naive MC estimator is known to incur very high variance; an easy remedy is to reweight \eqref{eqn:fp_def} with another density $q$ to obtain 
\[
     P_F = \int_{\mcl{X}} \mathbbm{1}_{\{g(\x) > t\}}  w(\x) q(\x) \; d\x,
\]
where $w(\x) = \frac{p(\x)}{q(\x)}$ are the importance weights for the corresponding MC estimator, also known as the importance sampling (IS)~\cite{srinivasan2002importance}  estimator, given by:
\[
    P_F \approx \widehat P_F^\T{IS} = \frac{1}{M}\sum_{i=1}^M \mathbbm{1}_{\{g(X_i) > t\}} w(X_i)
    \quad\T{for}\quad X_i\sim q(\x), \;\; i=1,\dots, M,
\]
where $q$ is chosen such that it is either easier to sample from or has more desirable properties than $p$. If $q$ is chosen well, for example, to hold a high probability in the failure regions, then significant variance reduction can be achieved for $M \ll N$. On the other hand, a poor choice of $q$ can result in the variance of $\widehat P_F^\T{IS}$ exceeding that of $\widehat P_F^\T{MC}$. Therefore, choosing a good $q$ is crucial, but it is not straightforward because $g$ is an oracle with unknown structure. A surrogate model is commonly used to inform the estimation of $q$, see e.g., \cite{peherstorfer2016multifidelity, booth2024actively, renganathan2023camera,li2011efficient,dubourg2013mbis}.

The variance of the IS estimator is given as
\[
\mathrm{Var}_q\!\left(\widehat P_F^{\text{IS}}\right)
= \frac{1}{M}\,
\mathrm{Var}_q\!\left(\mathbbm{1}_{\{g(X)>t\}}\,w(X)\right)
= \frac{1}{M}\left(
\int_{\mathcal X} \mathbbm{1}_{\{g(\x)>t\}}
\frac{p(\x)^2}{q(\x)}\,d\x
\;-\; P_F^2
\right).
\]
Then, it can be shown that the \emph{optimal} IS density $q^*$ is the one that results in zero variance of the estimator $\widehat P_F^\T{IS}$ and is given as
\[
q^{\star}(\mathbf{x})
\;=\;
\frac{\mathbbm{1}_{\{g(\mathbf{x})>t\}}\,p(\mathbf{x})}{P_F}.
\]
Naturally, the optimal density $q^\star$ is impossible to estimate unless we know $P_F$ itself. However, a density that is $\propto \mathbbm{1}_{\{g(\mathbf{x})>t\}} p(\x)$ serves as a good target. Although we don't know $\mathbbm{1}_{\{g(\mathbf{x})>t\}}$, a consistent approximation of it could be very fitting.

The proposal density $q$ can be chosen with the help of a surrogate model. Specifically, if $g$ is approximated with a surrogate model $\hat{g}$, then $\hat{g}$ can, in turn, be used to approximate the set $F$, which in turn maybe used to inform the choice of $q$, e.g., using kernel density estimation~\cite{ TerrellScott1992,Tsybakov2009,Botev2010}. There are several works from the past decade that use this approach: first, construct a surrogate model $\hat{g}$ for $g$ using observations $g(\x_i),~i=1, \ldots, n$; second, use $\hat{g}$ to propose a $q$; and finally, compute $\widehat P_F^\T{IS}$ using separate oracle evaluations sampled from $q$. We call such an approach ``two-stage,'' due to the two disconnected stages: constructing a surrogate and then estimating $P_F$. The main drawback of the two-stage approach is that expensive evaluations of $g$ used to train the surrogate are not reusable for estimating $P_F$ because the surrogate $\hat{g}$ is generally fit with a global approximation goal. It is likely that most of the evaluations of $g$ used to train $\hat{g}$ are not in $F$ for them to be useful in estimating $P_F$.

In the conventional two-stage approach, it is often argued that the central burden lies in building an accurate surrogate of the limit state, while the construction of the biasing density $q$ is treated as secondary~\cite{peherstorfer2016multifidelity}. In this regard, oracle evaluations are prioritized for surrogate-based active learning of the failure boundaries.  Then, remaining evaluations are sampled (hopefully in $F$) using the surrogate-informed $q$.
We take the opposite view: \emph{the biasing density is paramount for accurately estimating $P_F$ and should be prioritized}. In this regard, we aim to optimally choose oracle evaluations that serve both the surrogate training and fitting $q$. Indeed, if one had access to the optimal (zero-variance) IS density $q^*$, a single sample suffices to recover the exact failure probability. We briefly formalize this statement and illustrate it with a two-dimensional toy example.
\begin{proposition}
\label{prop:optimalq}
If $X_1,\dots, X_n \stackrel{\mathrm{i.i.d.}}{\sim} q^\star$, then for every $n\ge 1$,
\[
\widehat P_F^{\rm IS} = P_F \quad \text{almost surely}.
\]
In particular, the estimator has zero variance, and a single sample suffices to obtain the exact value of $P_F$.
\end{proposition}

\begin{proof}
Under $q^\star$, $ X_i\in F$ almost surely, hence $\mathbbm 1_F( X_i)=1$ a.s. Moreover,
\[
w(X_i)
=
\frac{p(X_i)}{q^\star( X_i)}
=
\frac{p( X_i)}{p( X_i)\,\mathbbm 1_F(X_i)/P_F}
=
P_F \quad \text{a.s.}
\]
Therefore, each summand in the IS estimator equals $P_F$ a.s., so their average equals $P_F$ a.s.; hence, the variance is $0$.
\end{proof}

\begin{figure}[htb!]
    \centering
    \includegraphics[width=1\textwidth]{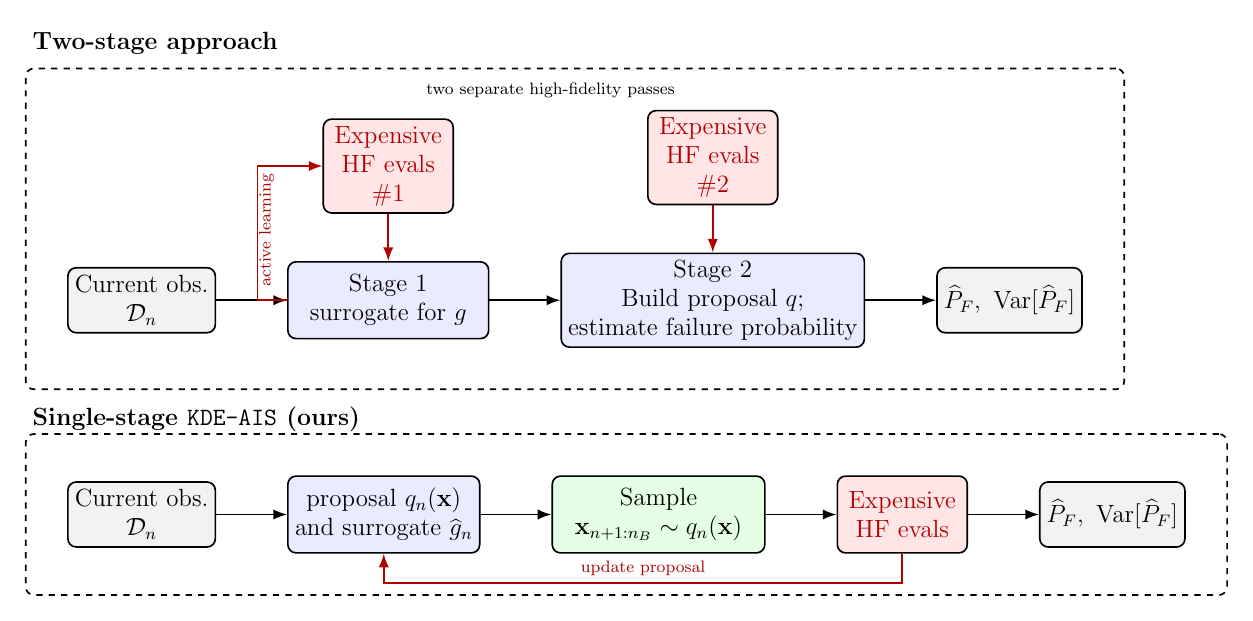}
    \caption{Overview of the proposed ``single-stage'' approach where high fidelity (HF) oracle evaluations inform both the biasing density and the surrogate, in contrast to existing ``two-stage'' approaches.}
    \label{fig:overview}
\end{figure}

In this work, we seek to emulate $q^\star$ as opposed to emulating $g$~\cite{renganathan2023camera} or contours of $g(\x) = t$~\cite{booth2025two, booth2025contour}; in the process, however, we show that $g(\x) = t$ is also accurately emulated. We develop sequential approximations that are guaranteed to recover $q^\star$ asymptotically -- this is popularly known as adaptive importance sampling (AIS)~\cite{bugallo2017adaptive,oh1992adaptive}. Crucially, we
take a \emph{single-stage} approach, where a surrogate of the limit state $\hat{g}$ and the estimate $\widehat P_F$ are obtained using the \emph{same} sample evaluations of $g$; this way, we hope to 
accurately estimate $P_F$ with substantially fewer evaluations of $g$ compared to two-stage approaches.
The idea is to sequentially update both $\hat{g}$ and $\hat{q}$ using these evaluations, which then leads to a sequential update to $\widehat P_F$. Our objective in this process is to ensure that both $\hat{g}$ and $\hat{q}$ are consistent with $g$ and $q^\star$, respectively, and that the estimate $\widehat P_F$ has diminishing variance. We employ kernel-based methods, specifically Gaussian process (GP)~\cite{rasmussen:williams:2006,gramacy2020surrogates} models and kernel density estimation (KDE), as choices to learn the surrogate $\hat{g}$ and the biasing density approximation $\hat{q}$, respectively. \Cref{fig:overview} provides a schematic of our proposed method, contrasting it against the conventional two-stage approach.

\subsection{Related work}
\label{sec:related}
Classical reliability methods such as the first-order reliability method (FORM) and the second-order reliability method (SORM) are computationally efficient, but they have several well-known limitations. Both belong to the class of ``local reliability methods'' that first transform the basic random variables into a standard normal space and then make first/second order polynomial approximations of the limit-state~ \cite{madsen2006methods,huang2018reliability, huang2017overview,zhao1999response}. As a consequence, their accuracy deteriorates when the limit state is highly nonlinear, nonconvex, multimodal, or exhibits strong interactions among variables \cite{der2000geometry,hu2021second}. Crucially, they often require gradient and Hessian information of the limit-state function, which may be prohibitive in several real-world applications \cite{madsen2006methods,zhao1999response}. 

The limitations of classical FORM/SORM methods are easily overcome by surrogate-based methods. To be useful in estimating failure probabilities, a surrogate must be trained to accurately identify failure boundaries (i.e., contour location).  Seminal work on adaptive GPs for contour finding was performed by \cite{ranjan2008sequential} and \cite{bichon2008efficient}, who used \cite{jones1998efficient}'s expected improvement framework. 
Others have leveraged stepwise uncertainty reduction \cite{bect2012sequential,chevalier2014fast,azzimonti2021adaptive,duhamel2023version}, predictive entropy \cite{marques2018contour,cole2021entropy}, weighted prediction errors \cite{picheny2010adaptive}, or distance to the contour \cite{gotovos2013active} to target failure contours. Then, leveraging the ``two-stage'' framework, unbiased estimates of $P_F$ are obtained via importance sampling using additional evaluations of the expensive oracle. For instance, \cite{dubourg2013mbis} uses an adaptive surrogate model along with importance sampling to construct a quasi-optimal biasing distribution. \cite{peherstorfer2016multifidelity} proposed using a surrogate to identify inputs from $p(\x)$ that are predicted to fail, then fitting a Gaussian mixture model \cite{reynolds2015gaussian} to those locations for the biasing density.  
Several other surrogate assisted IS approaches, e.g., ~\cite{dubourg2013mbis,balesdent2013akis,cadini2014improvedAKIS}, and refinements with stratified/directional IS, system reliability, mixture fitting, and multiple importance sampling (MIS) reuse~\cite{persoons2023akis,xiao2020aksis} also exist. Another notable line of work is subset simulation, which breaks $P_F$ into products of larger conditional probabilities~\cite{au2001ss}; this idea has also been combined with active learning~\cite{huang2016akss,zhang2019akss,bect2017bayesSubSim}. Recent work separates surrogate and sampling errors and offers stopping rules ~\cite{booth2025two}. Yet these ``two-stage'' approaches are limited by their disjoint use of expensive evaluations for estimation of $\hat{g}$ and $\hat{q}$ (\Cref{fig:overview}), and may end up costing several thousands of oracle evaluations to accurately estimate $P_F$.

On the density estimation side, beyond parametric Gaussian/mixture proposals, \emph{nonparametric} and \emph{learned} transport importance sampling have been increasingly explored. Classic work estimated near-optimal IS densities by kernel density estimators from pilot samples, with unbiasedness and efficiency characterizations~\cite{ang1992kernel}. In reliability, AIS schemes with kernel proposals and Markov chain Monte Carlo exploration of failure regions have been attempted~\cite{au1999aiskernel,lee2017kdeis}. Nonparametric IS shows strong performance in rare events~\cite{li2021npis}. Separately, normalizing-flow proposals learn flexible transports toward failure sets~\cite{dasgupta2024rein}. Our proposed approach seamlessly integrates with any of these density estimation methods; however, we chose kernel density estimation to prove consistency results on our proposal.

Our approach closely resembles the ``GP adaptive importance sampling (GPAIS)'' approach by \cite{dalbey2014gpais}, where a GP surrogate approximation is used for $g$ to build an estimate of $q^\star$, but our contributions offer several notable improvements. Whereas GPAIS parametrizes the proposal directly from GP exceedance/expected-indicator quantities, we use the GP only to produce soft failure probabilities and then fit a separate nonparametric density model for the proposal. Second, GPAIS lacks any built-in mechanism that guaranties the exploration of $\mcl{X}$, and hence can miss isolated failure regions if the seed samples to the GP are not ``diverse'' enough.  In contrast, our \texttt{KDE-AIS} proposal is guaranteed to densely sample $\mcl{X}$, and therefore won't miss any failure regions. Third, the theoretical guaranties in \texttt{KDE-AIS} extend beyond the unbiasedness and lower variance of the MIS estimator from GPAIS. We show that our proposal recovers the optimal $q^\star$, and hence our estimator converges to the true $P_F$ asymptotically. Crucially, GPAIS cannot offer this guarantee because their sampling is not guaranteed to be dense.
Fourth, unlike GPAIS, \texttt{KDE-AIS} uses deterministic-mixture MIS over the full history of proposals and then adds an explicit multifidelity (MF-MIS) estimator. Finally, we show that \texttt{KDE-AIS} performs empirically better than GPAIS in our experiments.

\subsection{Contributions}
Addressing the aforementioned gaps in the literature, our contributions are summarized as follows.

\begin{enumerate}
    \item We introduce a GP surrogate combined with a smoothing parameter $\alpha$ to construct a continuously evolving proxy target $q_n^\dagger$, which guards against surrogate bias during early iterations and promotes early exploration.
    \item We introduce a proposal $q_n$ that combines $q_n^\dagger$ and the input density $p$ using an exploration parameter $\eta$; this ensures that, asymptotically, the domain $\mcl{X}$ is densely sampled. This is a stark improvement over GPAIS.
    \item In addition to unbiasedness, our estimator is endowed with two notable features: (i) a complete reuse of all past proposals to $q$ via a balance heuristic and (ii) a multifidelity extension (MF-MIS) that shows improved sample efficiency compared to a traditional MIS estimator and has provably lower variance (\Cref{lem:mfmis-var-comparison}), as long as the surrogate evaluations are not too negatively correlated with the oracle evaluations.
    \item We show the following theoretical results (\Cref{thm:prop-conv}).
    \begin{enumerate}
        \item \label{result1} Our proxy target $q_n^\dagger$ has bounded error with $q^\star$ in total variation, which vanishes asymptotically.
        \item \label{result2} Under mild conditions on the exploration parameter $\eta$, our weighted proposal $q_n$ converges to $q_n^\dagger$ asymptotically while guaranteeing perfect emulation of $g$.
        \item Results in \ref{result1} and \ref{result2} are independent of the choice of the density estimation method. When a KDE approximation $\widehat q_n$ is used for $q_n$, we show that this approximation error also asymptotically vanishes. 
        \item As a consequence of \ref{result1} and \ref{result2}, we show that our estimate of $\widehat P_F$ asymptotically converges to $P_F$ with zero variance.
    \end{enumerate}
    \item Empirically, we demonstrate that our approach has improved sample efficiency and lower variance compared to several state-of-the-art methods, based on synthetic and real-world experiments.    
\end{enumerate}

The rest of the article is organized as follows. We provide the mathematical background on our methods in \Cref{sec:background} followed by the details of our method in \Cref{sec:method}; we discuss theoretical properties of our method in \Cref{sec:theory}. We demonstrate our method on synthetic and real-world experiments in \Cref{sec:experiments} and provide concluding remarks in \Cref{sec:conclusions}.

\section{Background}
\label{sec:background}

\subsection{Gaussian process surrogates}
\label{ss:GP}
The primary ingredient of our method is a Gaussian process surrogate model for the limit state function $g$.
 Denote observations of $g$ as $y_i = g(\x_i),~i=1,2,\ldots, n$. 
Let $\mbf X_n$ denote the stack of $n$ rows of $\x_i^\top,~i=1,\ldots,n$.  Let $\y_n$ denote the corresponding response vector.  A GP model assumes a multivariate normal distribution over the response, e.g., $\y \sim \mcl{GP}(0, \Sigma(\mbf X))$, where the covariance function $\Sigma(\mbf X)$ captures the pointwise correlations among observed locations, and
is typically a function of Euclidean distances, i.e., $\Sigma(\mbf X)^{ij} = k(||\x_i - \x_j||^2)$; see \cite{santner2003design,rasmussen:williams:2006,gramacy2020surrogates} for reviews.  Conditioned on observations $\D_n = \{\mbf X_n, \y_n\}$, the posterior predictive distribution at input $\x$ is also Gaussian and follows 
\begin{equation}
    Y_n(\x) | \D_n \sim \mcl{GP}\left(\mu_n(\x), \sigma^2_n (\x)\right)
    \quad\T{where}\quad
    \begin{aligned}
    \mu_n(\x) &= \Sigma(\x, \mbf X_n)\Sigma(\mbf X_n)^{-1}\y_n \\
    \sigma^2_n(\x) &=  \Sigma(\x) - \Sigma(\x, \mbf X_n)\Sigma(\mbf X_n)^{-1}\Sigma(\mbf X_n, \x).
    \end{aligned}
    \label{e:GP}
\end{equation}
Throughout, subscript $n$ is used to denote quantities from a surrogate trained on $n$ data points.  The posterior distribution of \Cref{e:GP} provides a general probabilistic surrogate model that can be used to approximate the limit state function. 

The uncertainty quantification provided by the GP facilitates Bayesian decision-theoretic updates to the surrogate model in a principled fashion -- popularly known as ``active learning''~\cite{santner2003design}.  Given an initial design $\D_n$ and some ``acquisition'' function $h(\x \mid \D_n)$ that quantifies the utility of a candidate input $\x$, the next input location may be optimally chosen as $        \x_{n+1} = \arg\max_{\x\in\X}~ h(\x \mid \D_n).$
The oracle is evaluated at $\x_{n+1}$, the data is augmented with $\{\x_{n+1}, g(\x_{n+1})\}$, the sample size is incremented to $n\leftarrow n+1$, and the process is repeated until the allocated budget is exhausted. This approach has been used for estimating failure probability with GPs: see, e.g., \cite{renganathan2023camera,ranjan2008sequential,bichon2008efficient,echard2011akmcs}.
In this work, our acquisitions are directly sampled from the current approximation for the biasing density $\hat{q}_n$, which circumvents any inner optimization.

\subsection{Adaptive importance sampling}\label{subsec:AIS}
\label{sec:ais}
Adaptive importance sampling refers to the adaptive improvement of the estimate of the biasing density in terms of reducing the variance of the importance sampling estimator~\cite{bugallo2017adaptive}. In this work, we make data-driven updates to a nonparametric approximation $\widehat q_k,~k=1,2,\ldots$ to the optimal (zero-variance) IS density $q^\star$. This, in turn, leads to an adaptively improving estimate of the failure probability. Specifically, we use a multiple importance sampling estimator that re-weights all samples up to the current iteration $k$ with a mixture denominator defined as follows. At iteration $k$ we draw $X_{k,i}\sim \widehat q_k$ ($i=1,\dots,n_k$). Then the current MIS estimator is given as
\[
\widehat{P_F}^{\text{MIS}}
= \frac{1}{N}\sum_{k=1}^K \sum_{i=1}^{n_k} 
\frac{p(X_{k,i})}{\bar{q}(X_{k,i})}\,\mathbbm 1(X_{k,i}),
\quad
\bar{q}(\mathbf{x}) \;=\; \sum_{j=1}^K \nu_j \widehat q_j(\mathbf{x}),\;\;
\nu_j=\frac{n_j}{N},\;\; N=\sum_{j=1}^K n_j,
\]
known as the \emph{deterministic mixture} or \emph{balance heuristic} \cite{VeachGuibas1995,ElviraMIS}. Deterministic mixture weights can substantially reduce weight variability compared to a naive IS estimator while retaining unbiasedness \cite{Owen2013,ElviraMIS}.

\subsection{Kernel density estimation}
We use kernel density estimation to develop an approximation $\widehat q_k$. KDE is a classical nonparametric approach to approximating an unknown density from samples \cite{rosenblatt1956,parzen1962,silverman1986,wandjones1995,scott2015}. 
In our setting, the \emph{biasing} (proposal) density is denoted by $q$ and the KDE approximation by $\widehat q$.
Let $X_1,\ldots,X_n \in \mathbb{R}^d$ be i.i.d.\ draws from $q$, and let $\mathbf{x}\in\mathbb{R}^d$ denote a point at which the density is evaluated. 
Given a kernel $K:\mathbb{R}^d\to\mathbb{R}$ with $\int K(\mathbf{u})\,\mathrm{d}\mathbf{u}=1$, $\int \mathbf{u}\,K(\mathbf{u})\,\mathrm{d}\mathbf{u}=\mathbf{0}$, and finite second moments, and a positive definite bandwidth matrix $\mathbf{H}\in\mathbb{R}^{d\times d}$, the multivariate KDE is
\[
  \widehat q(\mathbf{x})
  \;=\; \frac{1}{n}\sum_{i=1}^n K(\mathbf{x}-X_i),
  \qquad 
  K(\mathbf{u}) \;=\; |\mathbf{H}|^{-1/2}\, K\!\big(\mathbf{H}^{-1/2}\mathbf{u}\big).
\]
A common special case is the isotropic bandwidth $\mathbf{H}=h^2\mathbf{I}_d$ with scalar $h>0$, in which case
\[
  \widehat q(\mathbf{x})
  \;=\; \frac{1}{n h^d}\sum_{i=1}^n K\!\left(\frac{\mathbf{x}-X_i}{h}\right).
\]
Typical choices of $K$ include the Gaussian kernel $K(\mathbf{u})=(2\pi)^{-d/2}\exp(-\tfrac12\|\mathbf{u}\|_2^2)$ and compactly supported kernels, e.g., \cite{epanechnikov1969non, silverman1986kernel}.
Under the conditions $h\to0$ and $n h^d\to\infty$, $\widehat q(\mathbf{x})\to q(\mathbf{x})$ pointwise and in $L_2$ \cite{silverman1986,wandjones1995}.

The bandwidth parameter (or more generally, bandwidth matrix) dominates performance; the precise kernel choice is much less important \cite{silverman1986,wandjones1995,scott2015}.  We use the ``normal-reference'' rule of thumb in selecting the bandwidth parameter. If $q$ is approximately Gaussian with covariance $\boldsymbol{\Sigma}$, a convenient full-matrix choice is
  \begin{equation*}
    \mathbf{H}_{\mathrm{NR}} 
    \;=\; \left(\frac{4}{d+2}\right)^{\!\frac{2}{d+4}} n^{-\frac{2}{d+4}} \,{\boldsymbol{\Sigma}},
  \end{equation*}
  which reduces in $d=1$ to Silverman’s rule $h_{\mathrm{NR}} \approx 1.06\,\hat\sigma\,n^{-1/5}$ \cite[Sec.~6.2]{silverman1986,scott2015}.

\section{\texttt{KDE-AIS:} Kernel density estimation adaptive importance sampling}
\label{sec:method}
We now present our method, which combines adaptive surrogate modeling with GPs, density estimation with KDEs, and multiple importance sampling.

\subsection{Proposal density $q$ estimation}
\label{sec:q_estimation}
The first step is the surrogate-based IS density estimation. Given data $\mathcal D_n=\{(\x_i,y_i)\}_{i=1}^n$ with $y_i=g(\x_i)$, we fit a GP and denote its posterior mean and variance as $\mu_n,\ \sigma_n^2$. The \emph{surrogate probability of failure} at $\x$ is then given by
\[
\pi_n(\x)\;=\;\Pr\!\big(Y>t\mid\mathcal D_n\big)
\;=\;1-\Phi\!\Big(\tfrac{t-\mu_n(\x)}{\sigma_n(\x)}\Big),
\]
which follows from the Gaussianity of $Y$. We use $\pi$ to guide an ``evolving'' target density. Recall that the optimal (that is, zero-variance) proposal for importance sampling is given as $q^\star(\x)\propto p(\x)\mathbbm 1_F(\x)$. We argue that using $\pi_n$ as a plug-in replacement for $\mathbbm 1_F(\x)$ is quite appropriate because, as we show later, $\lim_{n \rightarrow \infty} \pi_n(\x) = \mathbbm 1_{F}(\x)$. However, instead of setting the target as $\propto p(\x) \pi_n(\x)$, we propose a ``smoothed'' proxy target defined as
\[
q_n^\dagger(\x) \;\propto\; p(\x)\,\Big[\pi_n(\x)\Big]^\alpha,\qquad \alpha\in(0,1].
\]
Note that when $\alpha = 0$, we recover the standard MC estimate. This smoothing is done for the following reasons. First, when $\alpha=1$, we place complete belief on the surrogate estimate of the failure region, which could lead to erroneous estimates during early stages when the surrogate is expected to be biased. $\alpha<1$, on the other hand, guards against surrogate errors and promotes exploration early on. Ideally, we want to explore when the surrogate is less confident and exploit when the surrogate is more confident. Second, the importance weights $p(\x)/q_n^\dagger(\x) \propto \pi_n(\x)^{-\alpha}$ -- therefore, $\alpha=1$ could blow up these weights when $\pi_n(\x) \approx 0$ and $\alpha<1$ guards against that. Finally, we show later in Theorem 1 that regardless of the choice of $\alpha \in (0,1]$, our target density is still consistent.

We estimate the target density $q_n^\dagger$ via KDE.  For a set of draws $\{X_j\}_{j=1}^m\stackrel{{\rm iid}}{\sim}p$, and given $\pi_n$, define weights $w_j=\big[\pi_n(u_j)\big]^\alpha$ and normalize as
$\tilde w_j=w_j/(\sum_{k=1}^m w_k)$, $\forall\, j = 1,\ldots,m$. For bandwidth $h>0$, form the weighted Gaussian KDE $\widehat{q}_n$ as (we illustrate in $1$D for simplicity)
\[
\widehat q_n(\x)\;=\;\sum_{j=1}^m \tilde w_j\,\varphi_h\big(\x-X_j\big),
\qquad
\varphi_h(z)=\frac{1}{(2\pi h^2)^{d/2}}\exp\!\Big(-\tfrac{\|z\|^2}{2h^2}\Big),
\]
where $\varphi$ is the Gaussian kernel with bandwidth $h$. Note that $\widehat{q}_n$ approximates $q_n^\dagger$ -- we show (in Theorem 1) that the associated approximation error is bounded for finite $n$ and vanishes as $n \rightarrow \infty$. One potential issue with $\widehat{q}_n$ is that it still depends on the accuracy of the surrogate estimate of failure regions. There is a nontrivial chance that a failure region, initially missed by the surrogate, can go undetected in the limit. To circumvent this pathology, we introduce an exploration parameter $\eta \in (0,1)$ which combines the KDE with the input density $p(\x)$ and is given as
\[
q_n(\x)\;=\;(1-\eta_n)\,\widehat q_n(\x)\;+\;\eta_n\,p(\x),
\]
with $\eta_n\in(0,1)$ decaying slowly to $0$. Under some conditions on $\eta_n$, we show that this guarantees exploration and will result in an asymptotically dense sampling on $\mcl{X}$. 

\subsection{GP and failure probability updates}
\label{sec:gp_fp_updates}
After iteration $n$, unlike Bayesian decision-theoretic active learning with GPs, a batch of $N_b$ new acquisitions are directly sampled from $q_n$:
\[\x_{n+1:N_b} \sim q_n(\x).\]
That is, our acquisition does not depend on solving another ``inner'' optimization problem typical of Bayesian decision-theoretic approaches, but directly samples from the current proposal $q_n$. Sampling from $q_n$ is straightforward and involves two steps. For every new sample, first draw from a Bernoulli distribution with probability $1-\eta_n$: $B\sim{\rm Bernoulli}(1-\eta_n)$. If $B=1$, then we sample from the KDE branch ($\widehat q_n$); if $B=0$, we sample from $p(\x)$. This ensures that (and later proved in \Cref{thm:prop-conv}) our sampling scheme is asymptotically dense, unlike other existing methods such as GPAIS.

\subsubsection{A simple multifidelity estimator}
\label{sec:mfmis}
When aggregating \emph{all} evaluations collected up to iteration $n$, we form a MIS estimator via the \emph{balance heuristic}.  Let $N_{\rm tot}=N_0+\sum_{k=1}^{n}N_k$ be the total number of evaluations of $g$ so far, and $q_k$ the proposal used at iteration $k$ (with $q_{0}\equiv p$ for the $N_0$ initial seed points).  Define the empirical mixture density
\[
    \bar q_{N_{\rm tot}}(\x)\;=\;\frac{N_0\,p(\x)\;+\;\sum_{k=0}^{n} N_k\,q_k(\x)}{N_{\rm tot}}.
\]
Then, the MIS estimate of the failure probability is
\[
\widehat P_{F, n}^{\,\rm MIS}
\;=\;\frac{1}{N_{\rm tot}}\sum_{i=1}^{N_{\rm tot}}
\mathbbm 1_{\{g(\x_i)>t\}}\;\frac{p(\x_i)}{\,\bar q_{N_{\rm tot}}(\x_i)\,},
\]
which is unbiased and typically exhibits reduced variance relative to weighting only by the proposal that generated each $\x_i$~\cite{dalbey2014gpais}. However, we are interested in accurately estimating $P_F$ with as few as $100$'s of evaluations of $g$. This can be challenging (as revealed by our experiments) since biases in the surrogate can, in turn, bias $\bar q_{N_{\rm tot}}$, leading to inaccurate $P_{F, n}^{\,\rm MIS}$.

To overcome this, we introduce a simple multifidelity estimator. At step $n$, let $\widehat g_n(\mathbf{x})$ denote the surrogate built from the
expensive evaluations collected up to that stage. Using the identity
\[
\mathbbm{1}_{\{g(\mathbf{x})>t\}}
=
\mathbbm{1}_{\{\widehat g_n(\mathbf{x})>t\}}
+
\Big(
\mathbbm{1}_{\{g(\mathbf{x})>t\}}
-
\mathbbm{1}_{\{\widehat g_n(\mathbf{x})>t\}}
\Big),
\]
the failure probability $P_F$ admits the exact decomposition
\[
P_F
=
\mathbb E_{\bar q}\!\left[
\mathbbm{1}_{\{\widehat g_n(\mathbf{x})>t\}}
\frac{p(\mathbf{x})}{\bar q(\mathbf{x})}
\right]
+
\mathbb E_{\bar q}\!\left[
\Big(
\mathbbm{1}_{\{g(\mathbf{x})>t\}}
-
\mathbbm{1}_{\{\widehat g_n(\mathbf{x})>t\}}
\Big)
\frac{p(\mathbf{x})}{\bar q(\mathbf{x})}
\right].
\]
Accordingly, we define the multifidelity MIS estimator
\[
\widehat P_{F,n}^{\,\mathrm{MF\text{-}MIS}}
=
\widehat P_{F,n}^{\,\mathrm{sur\text{-}MIS}}
+
\widehat P_{F,n}^{\,\mathrm{corr\text{-}MIS}},
\]
where
\[
\widehat P_{F,n}^{\,\mathrm{sur\text{-}MIS}}
=
\frac{1}{M_{\mathrm{tot}}}\sum_{i=1}^{M_{\mathrm{tot}}}
\mathbbm{1}_{\{\widehat g_n(\widetilde{\mathbf{x}}_i)>t\}}
\frac{p(\widetilde{\mathbf{x}}_i)}{\bar q_{M_{\mathrm{tot}}}(\widetilde{\mathbf{x}}_i)},
\]
and
\[
\widehat P_{F,n}^{\,\mathrm{corr\text{-}MIS}}
=
\frac{1}{N_{\mathrm{tot}}}\sum_{i=1}^{N_{\mathrm{tot}}}
\left[
\mathbbm{1}_{\{g(\mathbf{x}_i)>t\}}
-
\mathbbm{1}_{\{\widehat g_n(\mathbf{x}_i)>t\}}
\right]
\frac{p(\mathbf{x}_i)}{\bar q_{N_{\mathrm{tot}}}(\mathbf{x}_i)}.
\]
Here, $\{\widetilde{\mathbf{x}}_i\}_{i=1}^{M_{\mathrm{tot}}}$ denotes a large
surrogate-only MIS sample, while $\{\mathbf{x}_i\}_{i=1}^{N_{\mathrm{tot}}}$
are the expensive oracle evaluations available up to the current step $n$.
We set $M_{\mathrm{tot}} \gg N_{\mathrm{tot}}$, which is affordable because $\widehat P_{F,n}^{\,\mathrm{sur\text{-}MIS}}$ is independent of any oracle evaluations and hence is inexpensive to compute.

If the surrogate-only sample is generated using the same MIS mixture proportions
as the expensive sample, then
\[
\bar q_{M_{\mathrm{tot}}}(\mathbf{x}) = \bar q_{N_{\mathrm{tot}}}(\mathbf{x}),
\]
and the estimator simplifies to
\[
\widehat P_{F,n}^{\,\mathrm{MF\text{-}MIS}}
=
\underbrace{
\frac{1}{M_{\mathrm{tot}}}\sum_{i=1}^{M_{\mathrm{tot}}}
\mathbbm{1}_{\{\widehat g_n(\widetilde{\mathbf{x}}_i)>t\}}
\frac{p(\widetilde{\mathbf{x}}_i)}{\bar q_{N_{\mathrm{tot}}}(\widetilde{\mathbf{x}}_i)}
}_{\text{surrogate evaluations}}
+
\underbrace{
\frac{1}{N_{\mathrm{tot}}}\sum_{i=1}^{N_{\mathrm{tot}}}
\left[
\mathbbm{1}_{\{g(\mathbf{x}_i)>t\}}
-
\mathbbm{1}_{\{\widehat g_n(\mathbf{x}_i)>t\}}
\right]
\frac{p(\mathbf{x}_i)}{\bar q_{N_{\mathrm{tot}}}(\mathbf{x}_i)}.
}_{\text{oracle evaluations}}
\]
The first term is a cheap MIS estimate of the surrogate failure probability,
while the second term is a residual correction that removes the surrogate bias
using only the expensive oracle evaluations accumulated up to step $n$. Due to the unbiasedness of the MIS estimator, $\widehat P_{F,n}^{\,\mathrm{MF\text{-}MIS}}$ is also unbiased. The $P_{F,n}^{\,\mathrm{MF\text{-}MIS}}$ estimator is guaranteed to have a lower variance than the conventional MIS estimator, as long as $M_\T{tot} > N_\T{tot}$, and the surrogate and residual evaluation parts are not too negatively correlated. We formalize this in \Cref{lem:mfmis-var-comparison}.

\begin{lemma}[Conditions for variance reduction in MF-MIS]
\label{lem:mfmis-var-comparison}
Let $S_n$ and $R_n$ denote the surrogate and residual (oracle evaluations) contributions of $\widehat P_{F, n}^\T{MF-MIS}$, respectively. Let 
$V_{S,n}
:=
\operatorname{Var}\!\left(
\mathbbm{1}_{\{\widehat g_n(\mathbf{x})>t\}}
\frac{p(\mathbf{x})}{\bar q_{N_{\mathrm{tot}}}(\mathbf{x})}
\right)$ and 
$C_n
:=
\operatorname{Cov}\!\left(
\mathbbm{1}_{\{\widehat g_n(\mathbf{x})>t\}}
\frac{p(\mathbf{x})}{\bar q_{N_{\mathrm{tot}}}(\mathbf{x})},
\left[
\mathbbm{1}_{\{ g_n(\mathbf{x})>t\}}
-
\mathbbm{1}_{\{\widehat g_n(\mathbf{x})>t\}}
\right]
\frac{p(\mathbf{x})}{\bar q_{N_{\mathrm{tot}}}(\mathbf{x})}
\right).$
Then, 
\[
\operatorname{Var}\!\left(\widehat P_{F,n}^{\mathrm{MIS}} \right)
-
\operatorname{Var}\!\left(\widehat P_{F,n}^{\mathrm{MF\text{-}MIS}} \right)
=
\left(\frac{1}{N_{\mathrm{tot}}}-\frac{1}{M_{\mathrm{tot}}}\right)V_{S,n}
+\frac{2}{N_{\mathrm{tot}}}C_n.
\]
Consequently, if
\[
C_n \ge -\frac12\left(1-\frac{N_{\mathrm{tot}}}{M_{\mathrm{tot}}}\right)V_{S,n},
\]
then
\[
\operatorname{Var}\!\left(\widehat P_{F,n}^{\mathrm{MF\text{-}MIS}} \right)
\le
\operatorname{Var}\!\left(\widehat P_{F,n}^{\mathrm{MIS}} \right).
\]
\end{lemma}
\begin{proof}
    See \Cref{sec:mfmis_var_proof}.
\end{proof}
\subsection{Choosing parameters}
There are several parameters that need to be specified in our methodology, including $h$ (kernel bandwidth), $\alpha$ (smoothing exponent), and $\eta_n$ (exploration parameter); we now provide some guidelines for choosing them.
The bandwidth parameter of the KDE is chosen according to Silverman's rule of thumb~\cite{silverman1986}. 
Theoretically, the choice of the ``smoothing exponent'' $\alpha$ is insignificant; in practice, we recommend a default value of $\alpha = 0.97$ which worked well for all of our experiments.

 A critical choice is the exploration schedule $\eta_n$. We need $\sum_n\eta_n=\infty$ to ensure $p$ is sampled infinitely often -- this ensures a dense sampling in $\mcl{X}$ asymptotically and avoids pathologies like missing a failure region. We also need $\lim_{n \to 0} \eta_n = 0$ to ensure the density $q_n$ is asymptotically consistent -- this requires annealing $\eta_n$ to $0$.
Thus, we set the exploration schedule as follows:
\[\eta_n = \min \big\{1,\,c\,n^{-\gamma}\big\},\qquad 0<\gamma<1.\]
Since $\gamma$ is nonnegative, this sequence converges to $0$; further, $\sum_n n^{-\gamma} = \infty$ (since $\gamma < 1$).
The constant $c$ impacts convergence and other theoretical guarantees in this approach and thus must be chosen to keep the error, due to the KDE ($q_n$) and the surrogate failure probability ($\pi_n$), under control. In other words, $c$ must dominate the maximum of the error due to the KDE and the error due to the surrogate. In practice, we set $c=0.3$, which worked well for all the experiments conducted in this manuscript. However, the theoretical requirements behind the choice of $c$ are governed by the rate of decay of error in approximating $q^\star$ with $\widehat q_n$ -- this is discussed next.

The KDE error is due to two factors: the stochasticity in the samples used to fit the density and a bias term that stems from the KDE's modeling inadequacies, which are given as ~\cite{silverman1986}
\[ \|\widehat{q}_n - q^\dagger_n\| \approx \underset{\T{stochastic}}{\sqrt{\frac{\log m}{m h^d}}} + \underset{\T{bias}}{h^\beta}.\]
The surrogate error is the error in approximating the failure probability in $\mcl{X}$ -- this is defined as:
\[r_n = \| \pi_n(\x) - \mathbbm 1_{g(\x) > t}\|_{L^1(p)}.\]
Overall, we choose $c$ to ensure $\eta_n \gg \max(\|\widehat{q}_n - q^\dagger_n\|, r_n)$ because we want our exploration weight to decay slower than the errors in the KDE and surrogate, failing which we might end up not exploring $\mcl{X}$ while the surrogate is still not accurate enough. The natural question then is how can we estimate the surrogate error $r_n$, since the true indicator function is unknown. The following result provides an unbiased estimator $\widehat r_n$ of $r_n$ which can be estimated with the data $\D_n$ available at the current iteration.

\begin{proposition}[Unbiased estimator for $r_n$]
    \label{prop:estimator_for_rn}
    Recall that \(F=\{\x \in \mcl{X}:g(\x)>t\}\).
Then, the surrogate error is quantified as
\[
r_n \;=\; \mathbb E_p\!\big[\,|\pi_n(\x) -\mathbbm 1_F(\x)|\,\big]
      \;=\; \int_{\mathcal X} |\pi(\x)-\mathbbm 1_F(\x)|\,p(\x)\,d\x.
\]
And, an unbiased estimator for $r_n$ is given as
\[
\widehat r_n
\;=\;
\widehat P_F
\;+\;
\widehat{\mathbb{E}_p[\pi_n]}
\;-\;
2\,\widehat{\mathbb{E}_p[\pi_n\,\mathbbm{1}_F]},
\]
which can be estimated with no additional cost of evaluating the expensive limit state $g$ used to fit the GP surrogate.
\end{proposition}
\begin{proof}
    See \Cref{sec:estimator_for_rn}.
\end{proof}
The overall methodology is summarized in \Cref{alg:kde_mis_gp} - \Cref{alg:mis}. 
\begin{algorithm}[H]
\caption{\texttt{KDE-AIS:} Kernel density estimation adaptive importance sampling.}
\label{alg:kde_mis_gp}
\begin{algorithmic}[1]
\Require Input density $p(\mathbf{x})$ with a sampler; threshold $t$; pilot size $m$; initial size $n_0$; batch size $q$; iterations $T$; bandwidth $h>0$; exponent $\alpha\!\in\!(0,1]$; exploration schedule $\eta_n$ (e.g.\ $\eta_n=\min\{1,c\,n^{-\gamma}\}$); batch size $N_b$.
\State \textbf{Pilot:} draw $\{\mathbf{u}_j\}_{j=1}^m \stackrel{\text{iid}}{\sim} p$.
\State \textbf{Initialize data:} draw $\{ \mathbf{x}_i\}_{i=1}^{N_0}\!\sim\! p$, set $y_i=g(\mathbf{x}_i)$ and $\mathcal D_0=\{(\mathbf{x}_i,y_i)\}_{i=1}^{N_0}$.
\State Set $\text{counts}\leftarrow [\,N_0\,]$, $\text{proposals}\leftarrow[\text{``}p\text{''}]$, $N_{\text{tot}}\leftarrow N_0$.
\For{$n=0,1,\dots,T-1$}
  \State \textbf{Fit GP:} $\mathsf{GPfit}(\mathcal D_n)\ \to\ (\mu_n,\sigma_n^2)$.
  \State \textbf{Compute soft failure:} $\pi_n(\mathbf{u}_j)=1-\Phi\!\big((t-\mu_n(\mathbf{u}_j))/\sigma_n(\mathbf{u}_j)\big)$ for $j=1{:}m$.
  \State \textbf{KDE weights:} $w_{n,j} \leftarrow \pi_n(\mathbf{u}_j)^\alpha$, \ \ $\tilde w_{n,j} \leftarrow w_{n,j}/\sum_{k=1}^m w_{n,k}$.
  \State \textbf{Exploration schedule:} $\eta \leftarrow \eta_{n+1}$.
  \State \textbf{Draw batch from mixture $\widehat q_n$:} $\{\mathbf{x}^{(k)}\}_{k=1}^{N_b} \leftarrow \mathsf{SampleMixture}(\eta,\tilde{\mathbf{w}}_n,\widehat q_n)$.
  \State \textbf{Evaluate:} $y^{(k)} \leftarrow g(\mathbf{x}^{(k)})$, \quad $\mathcal D_{n+1} \leftarrow \mathcal D_n \cup \{(\mathbf{x}^{(k)},y^{(k)})\}_{k=1}^{N_b}$.
  \State \textbf{Bookkeeping for MIS:} append ``$q_n$'' to $\text{props}$, append $q$ to $\text{counts}$, and set $N_{\text{tot}}\leftarrow N_{\text{tot}}+q$.
  \State {\bf Online estimate:} $\widehat P_{F,n}\leftarrow \mathsf{MISEstimator}(\mathcal D_{n+1},\text{proposals},\text{counts})$.
\EndFor
\State \Return final GP, $\mathcal D_T$, and $\widehat P_{F,T}$.
\end{algorithmic}
\end{algorithm}

\vspace{0.75em}
\begin{algorithm}[H]
\caption{$\mathsf{SampleMixture}(\eta,\tilde{\mathbf{w}},\widehat q_n)$}
\label{alg:sample_mixture}
\begin{algorithmic}[1]
\Require $\eta\in(0,1)$, normalized weights $\tilde{\mathbf{w}}=(\tilde w_j)_{j=1}^m$.
\State For each new point independently:
\State \quad Draw $B\sim \mathrm{Bernoulli}(1-\eta)$.
\State \quad \textbf{if} $B=1$ (KDE branch): draw $\x \sim \widehat q_n$.
\State \quad \textbf{else} (exploration branch): draw $\mathbf{x}\sim p$.
\State \Return the collected batch $\{\mathbf{x}\}$.
\end{algorithmic}
\end{algorithm}

\begin{algorithm}[H]
\caption{$\mathsf{MISEstimator}(\mathcal D,\text{props},\text{counts})$}
\label{alg:mis}
\begin{algorithmic}[1]
\Require $\mathcal D=\{(\mathbf{x}_i,y_i)\}_{i=1}^{N_{\text{tot}}}$; a list \text{props} of proposals used: first entry ``$p$'' (for the $n_0$ initial points), then $\widehat q_1, \widehat q_2, \dots$; matching counts in \text{counts}.
\State Compute the empirical mixture density
\[
\bar{q}_{N_{\text{tot}}}(\mathbf{x}) \;=\; \frac{\text{counts}[1]\cdot p(\mathbf{x})\;+\;\sum_{k\ge 2}\text{counts}[k]\cdot \widehat q_{k-1}(\mathbf{x})}{N_{\text{tot}}}.
\]
\State Estimate failure probability via $\widehat P_F^{\T{MIS}}$ or $\widehat P_F^{\T{MF-MIS}}$.
\State \Return $\widehat P_F$.
\end{algorithmic}
\end{algorithm}

\section{Theoretical properties}
\label{sec:theory}
The main theoretical result of \texttt{KDE-AIS} is to show that our surrogate estimate $q_n$ converges to the optimal $q^\star$ in total variation. This automatically guarantees asymptotic convergence of the proposed MF-MIS estimator to $\widehat P_F$, and the asymptotic vanishing of the estimator variance to zero (via \Cref{prop:optimalq}). Our results depend on the following mild assumptions listed below. 
\begin{assumption}[Input density]\label{as:p}
$p$ is bounded and bounded away from zero on a compact set containing the support of $q_n^\dagger$; moreover $p$ is $\beta$-H\"older, $\beta>0$.
\end{assumption}

\begin{assumption}[Surrogate accuracy]\label{as:pi}
$\|\pi_n-\mathbbm 1_F\|_{L^1(p)} \to 0$ as $n\to\infty$; denote $r_n=\|\pi_n-\mathbbm 1_F\|_{L^1(p)}$.
Consequently, $q_n^\dagger \to q^\star$ in total variation, where $q^\star(\x)\propto p(\x)\,\mathbbm 1_F(\x)$ is the zero-variance IS proposal.
\end{assumption}
Note that, under mild regularity conditions on the GP kernel, it can be shown that the GP posterior mean $\mu_n$ asymptotically converges to $g$. This has been proven previously; see, e.g., Theorem 1 in \cite{renganathan2023camera}.

\begin{assumption}[Bandwidth and sample size]\label{as:bandwidth}
As $m \rightarrow \infty$, $h_m\to 0$ and $m h^d/\log m \to \infty$.
\end{assumption}

\begin{assumption}[Weighted KDE regularity]\label{as:weightedKDE}
$K$ is bounded and Lipschitz, and the weights satisfy $0\le w_i\le 1$. Then, the weighted KDE inherits the uniform consistency rates of the standard KDE.
\end{assumption}
Assumption 4 is the natural analog of standard results on kernel density estimators with probability weights; see, for example, \cite{BuskirkLohr2005} and \cite{Breunig2008}.

\begin{assumption}[Exploration schedule]\label{as:eps}
$\sum \eta_n = \infty$ and $
\eta_n \to 0$. 
\end{assumption}
Note that, we want $\eta_n \to 0$ to ensure once we have a good enough surrogate for $q^\star$, we want to only sample from that. However, $ \sum_n \eta_n = \infty$ ensures that $\mcl{X}$ is sampled infinitely often, thereby avoiding pathologies that lead to pockets of $\mcl{X}$ being missed.

We now present the main theoretical result of our work.
\begin{theorem}[Proposal convergence]
\label{thm:prop-conv}
Let $p$ be bounded and bounded away from $0$ on a compact $\mathcal X\subset\mathbb R^d$ and $\beta$-H\"older
(Assumption~1). Let the GP surrogate yield $\pi_n$ with
$r_n=\|\pi_n-\mathbbm 1_{\{g>t\}}\|_{L^1(p)}\to 0$ (Assumption~2).
From pilot samples $\{u_j\}_{j=1}^{m_n}\stackrel{\rm iid}{\sim}p$ and bandwidth $h_n\downarrow0$ with
$m_n h_n^d/\log m_n\to\infty$ (Assumption~3), define the weighted KDE
\[
\widehat q_n(\x)=\sum_{j=1}^{m_n}\tilde w_{n,j}\,\varphi_{h_n}(\x-u_j),
\qquad
\tilde w_{n,j}\propto \big[\pi_n(u_j)\big]^\alpha,\ \ 0<\alpha\le1,
\]
and the surrogate target $q_n^\dagger(\x)\propto p(\x)\,[\pi_n(\x)]^\alpha$.
Assume $0\le \tilde w_{n,j}\le 1$ (Assumption~4). Let the exploration–mixture proposal be
\[
q_n(\x)=(1-\eta_n)\,\widehat q_n(\x)+\eta_n\,p(\x),
\]
with $\eta_n\to0$, $\sum_n \eta_n=\infty$, and
$\eta_n \gg h_n^{\beta}+\sqrt{\log(m_n)/(m_n h_n^d)}\ \vee\ r_n$ (Assumption~5).
Then, as $n\to\infty$ (allowing $m_n\to\infty$),
\begin{align*}
\|\widehat q_n - q_n^\dagger\|_\infty
&= O_p\!\Big(\sqrt{\tfrac{\log(m_n)}{m_n h_n^d}} + h_n^{\beta}\Big),\\
\|q_n - q_n^\dagger\|_{TV} &\xrightarrow{p} 0,\\
\|q_n^\dagger - q^\star\|_{TV} &\le C\, r_n^{\,\alpha}\ \xrightarrow{} 0,
\end{align*}
and hence $\|q_n - q^\star\|_{TV}\to 0$, where $q^\star(\x)\propto p(\x)\,\mathbbm 1_{\{g(\x)>t\}}$.
\end{theorem}
\begin{proof}
    See \Cref{sec:thm1_proof}.
\end{proof}

\section{Experiments}
\label{sec:experiments}
We benchmark the proposed method on two synthetic and two real-world experiments. In all experiments, we set $m=10^7$ (the pilot sample size drawn from $p$), $\alpha = 0.97$, $h=0.2$, and $c=0.3$. Each experiment is started with a set of $N_0$ seed samples, chosen uniformly at random from $\mcl{X}$, and replicated $10$ times. We compare the evolution of our estimator, $\widehat P_F^\T{MF-MIS}$, against $\widehat P_F^\T{MIS}$ (to demonstrate the benefit of the multi-fidelity estimator) and GPAIS~\cite{dalbey2014gpais}. We also include three additional ``two-stage'' benchmarks, which split the total budget of oracle evaluations in each experiment into two parts: one for surrogate fitting ($\hat g$) and one for $P_F$ estimation. For instance, ``Two-stage (30-70)'' means $30\%$ of all samples were used for surrogate fitting with $70\%$ for $P_F$ estimation.  This competitor emulates the two-stage approach in \cite{peherstorfer2016multifidelity}.
A summary of the experimental setting is shown in \Cref{tab:summary}; details on each experiment are presented in the following sections. 

\begin{table}[htb!]
\centering
\caption{Summary of experimental setting.}
\label{tab:summary}

\begin{tabular}{lcccc}
\toprule
\textbf{Experiment} & 
\textbf{Input dim.} & 
\textbf{Seed points $N_0$} & 
\textbf{Iterations $T$} & 
\textbf{Batch size} \\
\midrule

Herbie ($t=2.0,~2.122$) & $2$ & $5$ & $100$ & $5$ \\
Four branch ($t=2.0, 3.1$) & $2$ & $5$ & $100$ & $5$ \\
Cantilever beam & $4$ & $50$ & $50$ & $5$ \\
Shaft torsion & $5$ & $50$ & $200$ & $5$ \\

\bottomrule
\end{tabular}

\end{table}

\subsection{Synthetic experiments}
\subsubsection{Herbie function}
\label{sec:herbie}
As a first synthetic benchmark, we consider the Herbie test function \cite{lee2011herbie}, which has been used extensively in
reliability studies \cite{dalbey2014gpais,romero2016pof}. For
$\x = (x_1,x_2)^\top \in \mcl{X} = [-2,2]^2$, the limit state function
$g : [-2,2]^2 \to \mathbb{R}$ is defined as
\[
  g(\x)
  \;=\;
  \sum_{d=1}^{2}
  \bigg[
    \exp\!\big( -(x_d - 1)^2 \big)
    +
    \exp\!\big( -0.8 (x_d + 1)^2 \big)
    -
    0.05 \sin\!\big( 8 (x_d + 0.1) \big)
  \bigg].
\]
This function is smooth but highly multi–modal due to the superposition of two Gaussian-like bumps and an oscillatory term in each coordinate, making the resulting failure set disconnected and geometrically intricate.  We set $p$ to be uniformly distributed over $\mathcal{X}$.
\begin{figure}[htb!]
    \centering
    \begin{subfigure}{.5\textwidth}  
    \centering
    \includegraphics[width=1\linewidth]{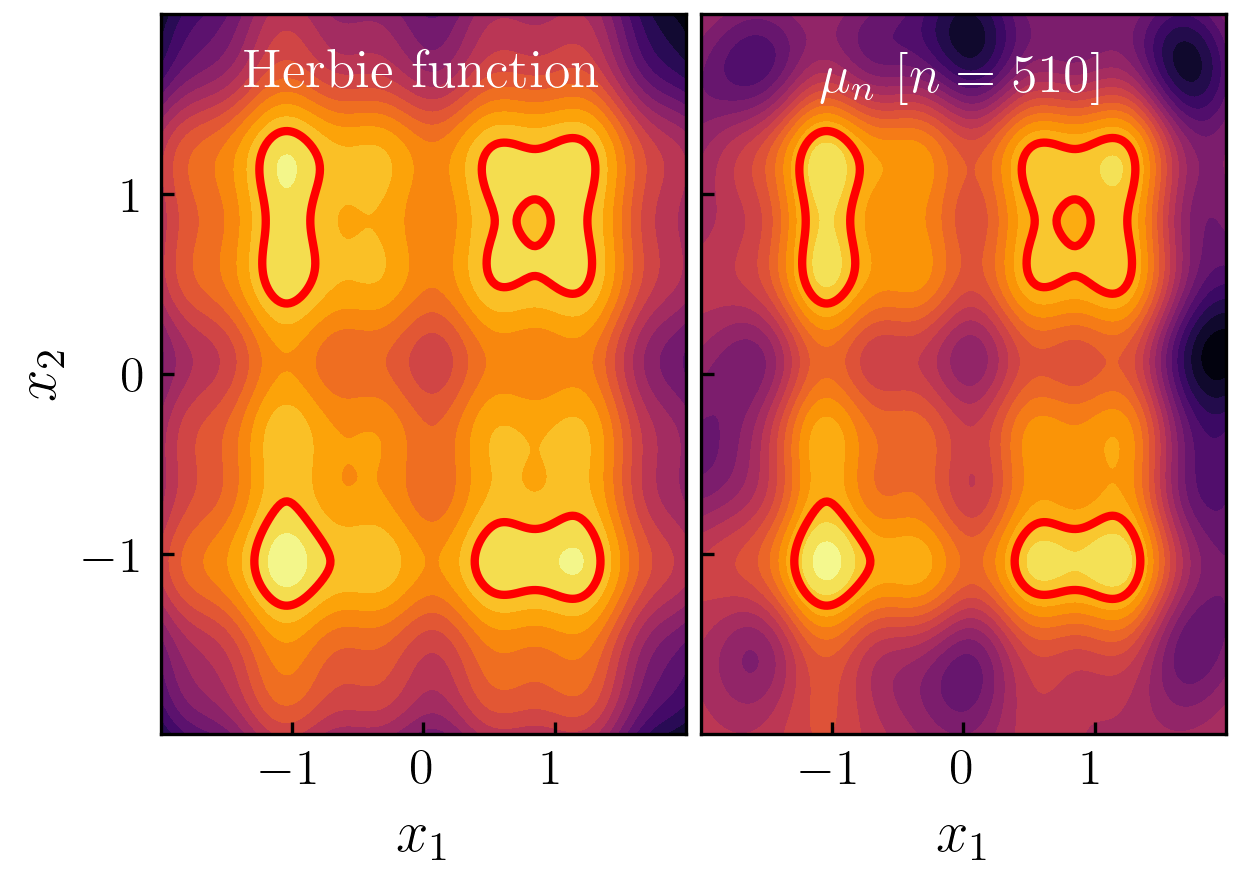}
    \end{subfigure}%
    \begin{subfigure}{.5\textwidth}  
    \centering
    \includegraphics[width=1\linewidth]{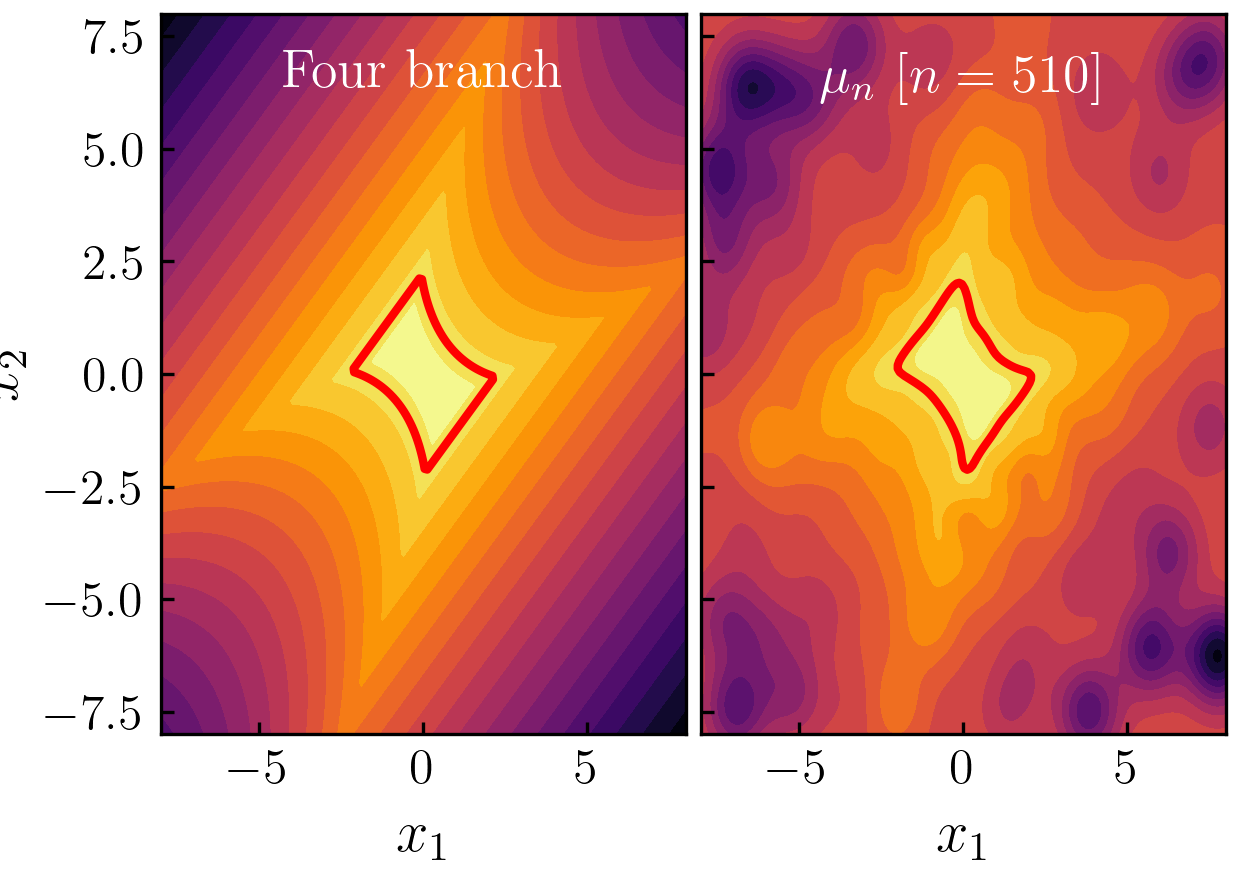}        
    \end{subfigure} \\
    \begin{subfigure}{.48\textwidth}
        \centering
        \includegraphics[width=1\linewidth]{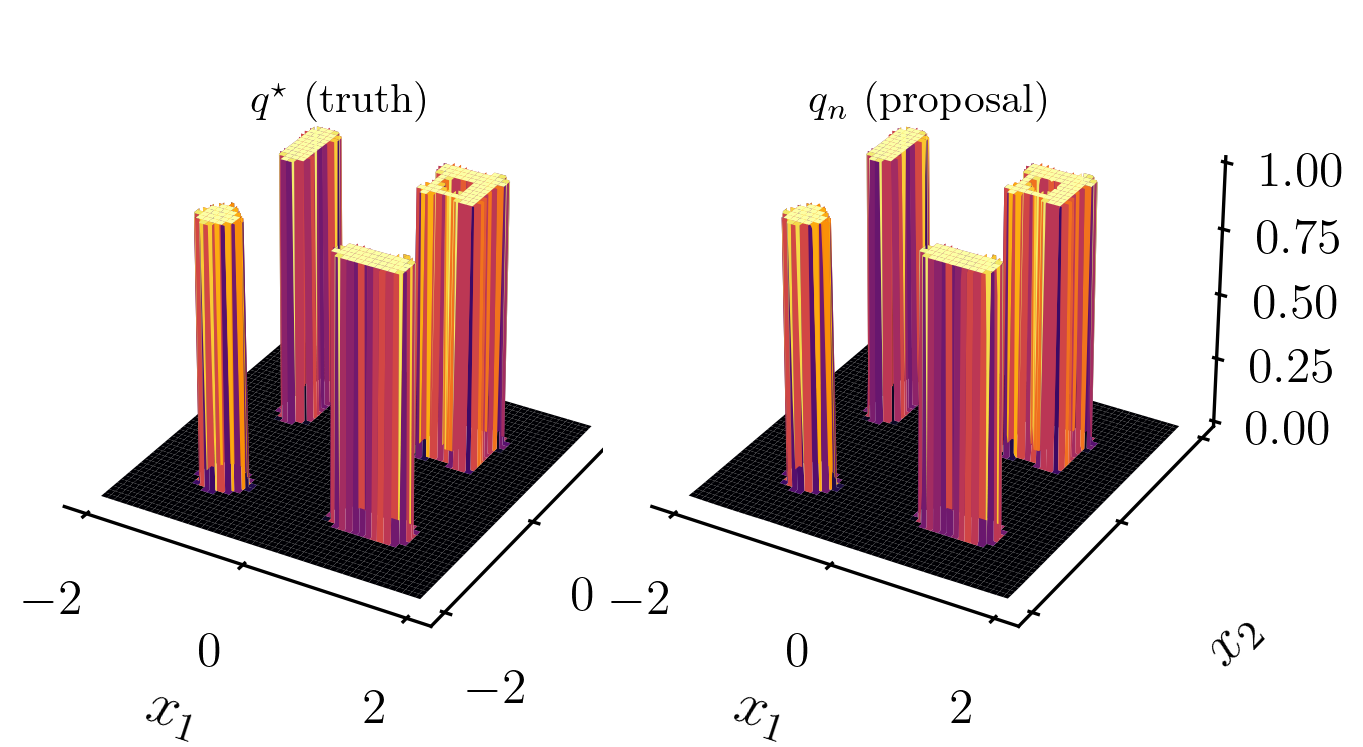}
        \caption{Herbie function}
    \end{subfigure}%
    \begin{subfigure}{.48\textwidth}
        \centering
        \includegraphics[width=1\linewidth]{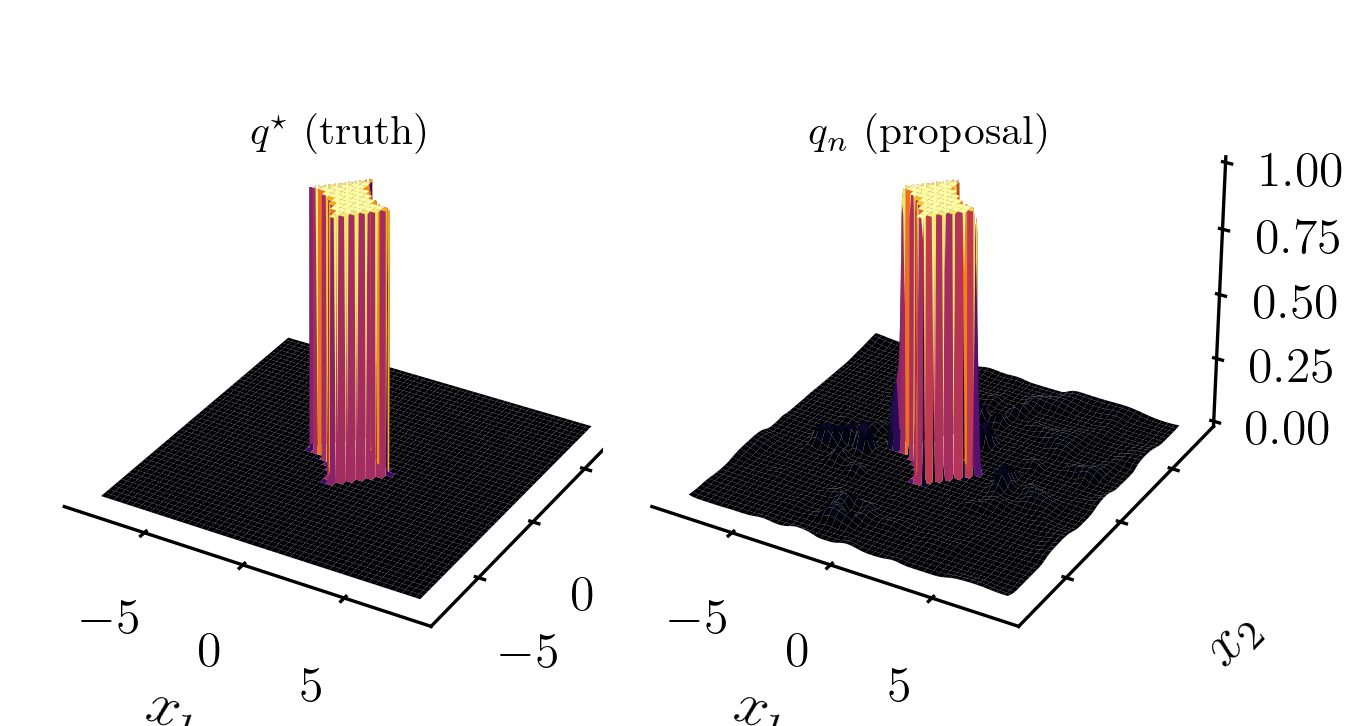}
        \caption{Four branch function}
    \end{subfigure}%
    \caption{\texttt{KDE-AIS} on the Herbie function (a) and four branch function (b). {\it Interior top left:} true function with failure boundaries as red lines. {\it Interior top right:} predicted surrogate failure region after $n=510$. {\it Interior bottom left/right:} optimal/predicted (after $n=510$) IS proposal densities.}
    \label{fig:synthetic_final}
\end{figure}

\Cref{fig:herbie_evolution} shows snapshots of the GP posterior mean, overlaid with seed (black) and acquisition (white) points, for various $n$; the red lines indicate the level set $\hat g_n(\x) = t$. Notice that the $N_0 = 5$ seed points miss $3$ out of the $4$ failure regions and, yet, the acquisitions explore the design space to identify all $4$ failure regions within $n=45$. With increasing $n$, the surrogate converges to the final prediction in about $150-200$ total samples. At $n=510$, there are several samples squarely within the failure regions which would, as shown next, enable accurate failure probability estimation.

We provide the comparison of the final prediction at $n=510$ against the truth in the left side of \Cref{fig:synthetic_final}. The top row shows the GP posterior mean $\mu_n$ (right) and the true $g$ (left), which are closely aligned. Additionally, the bottom row shows the predicted proposal $q_n$ (right) beside the true $q^\star$ (left) -- notice how closely $q_n$ emulates $q^\star$, substantiating the main proposal consistency result from \Cref{thm:prop-conv}. Crucially, $q^\star$ is nonsmooth; yet, our surrogate estimate, despite using smooth prior assumptions, is able to approximate it almost exactly.

The ultimate test of the method is in its ability to accurately estimate $P_F$. The top row of \Cref{fig:synthetic} shows the evolution of $\widehat P_F$ with the number of evaluations; we start with $5$ seed points and run the algorithm for $100$ additional iterations, with $5$ acquisitions each. We try two thresholds $t=2$ and $t=2.122$, which result in true failure probabilities of $0.00199$ and $9.144 \times 10^{-5}$, respectively. The proposed MF-MIS estimator has the most accurate estimate with the smallest variance compared to competing methods. Importantly, notice that the estimate settles down to the final value in $100$ evaluations for $t=2$ and $\sim 200$ evaluations for $t=2.122$ -- indicative of the sample efficiency of the proposed method. In comparison, GPAIS consistently overestimates, and {\tt KDE-AIS}-MIS consistently underestimates.  The two-stage benchmark shows consistently inaccurate estimates irrespective of the choice of the split in the total samples. In methods such as GPAIS, the lack of an automatic means to densely sample $\mcl{X}$ can potentially miss isolated failure regions. This, in turn, leads to incorrect predictions of $q$ and its normalizing constant, resulting in overpredicting $P_F$.
We attribute the accuracy of our MF-MIS estimator to the following reasons. (i) Our input density weighting in $q_n$ balances exploration and exploitation, leading to an accurate emulation of the failure boundaries within a few hundred oracle evaluations (see \Cref{fig:herbie_evolution}), (ii) The surrogate part of our estimator (first term in $\widehat P_{F,n}^\T{MF-MIS}$), with an accurate surrogate model and a very high $M_\T{tot}$, leads to an accurate estimate of $P_F$ with low variance. (iii) Finally, the residual part of our estimator (the second term in $\widehat P_{F,n}^\T{MF-MIS}$) corrects for any bias in the surrogate estimate, leading to an improved estimate of $P_F$. Overall, in addition to emulating the failure boundaries, learning an accurate $q_n$ (that emulates $q^\star$) is crucial to estimating $P_F$ accurately with small data.

\begin{figure}[p]
    \centering
    \includegraphics[width=1\linewidth]{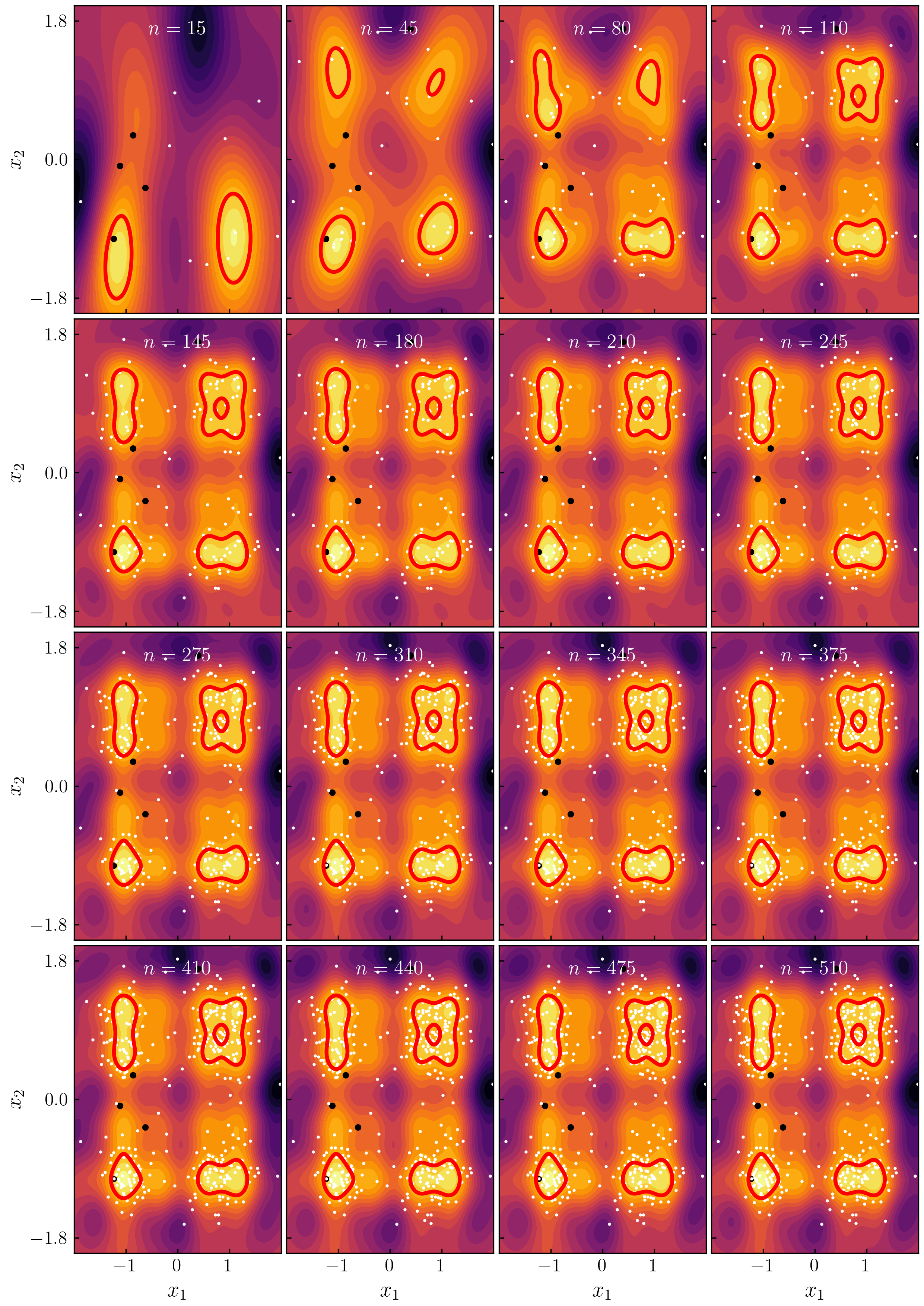}
    \caption{Snapshots of the posterior GP mean ($\mu_n$) for the Herbie function ($t=2.0$). Red lines represent learned limit sate $\widehat g = t$; black circles are $n=5$ seed points; white dots are samples drawn from the proposal $q_n$.}
    \label{fig:herbie_evolution}
\end{figure}

\begin{figure}[p]
    \centering
    \includegraphics[width=1\linewidth]{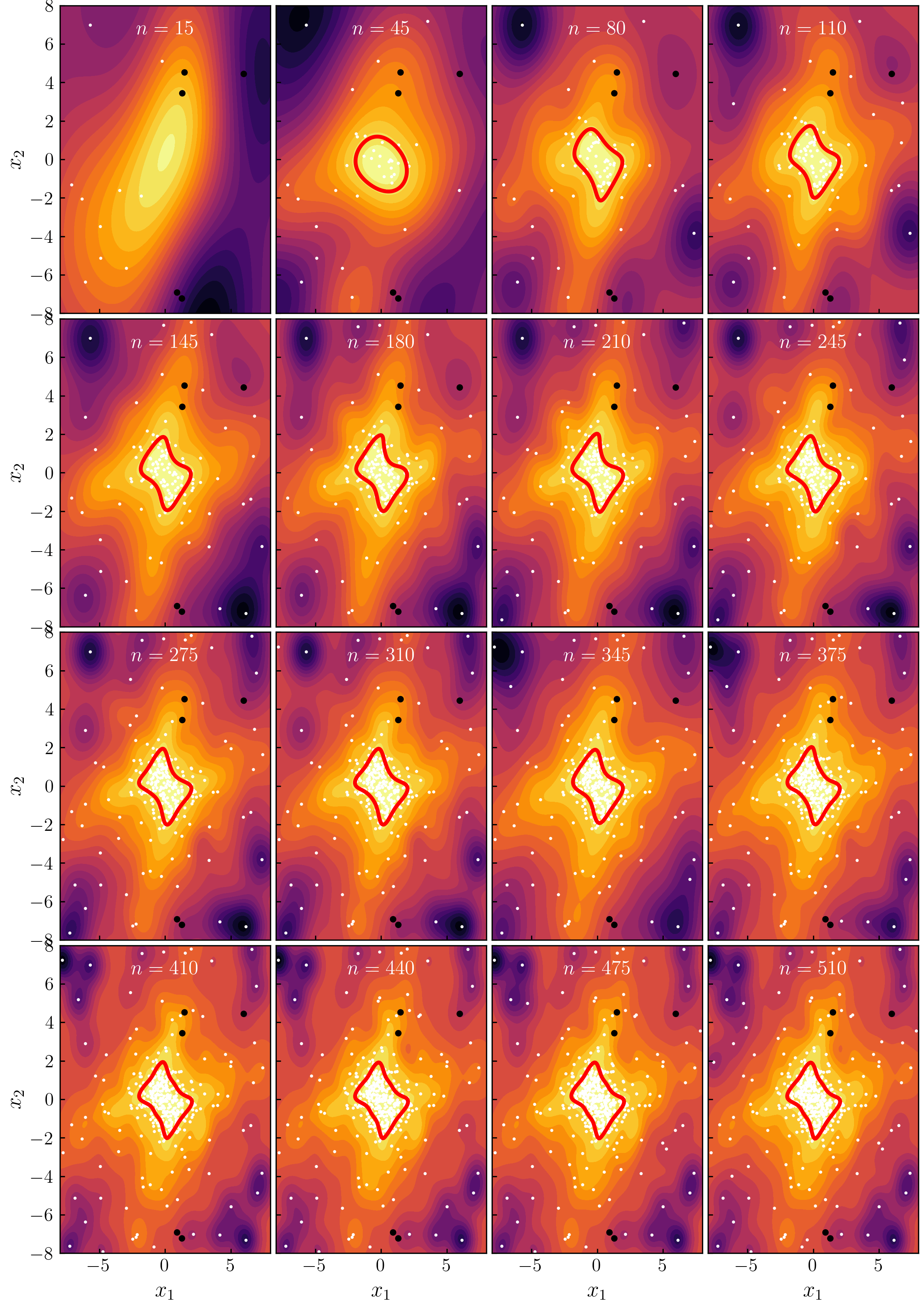}
    \caption{Snapshots of the posterior GP mean ($\mu_n$) for the Four branch function ($t=2.0$). Red lines represent learned limit sate $\widehat g = t$; black circles are $n=5$ seed points; white dots are samples drawn from the proposal $q_n$.}
    \label{fig:fb_evolution}
\end{figure}

\begin{figure}[htb!]
    \centering
    \begin{subfigure}{.5\textwidth}
        \includegraphics[width=1\linewidth]{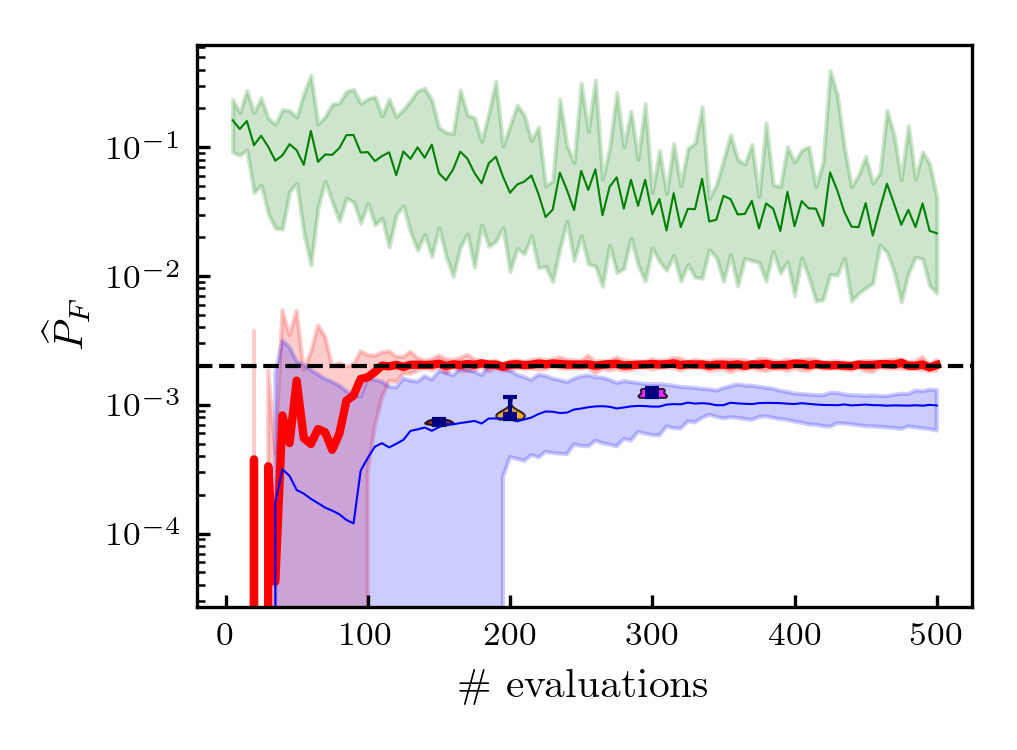}
        \caption{Herbie ($t=2.0$)}
    \end{subfigure}%
    \begin{subfigure}{.5\textwidth}
        \includegraphics[width=1\linewidth]{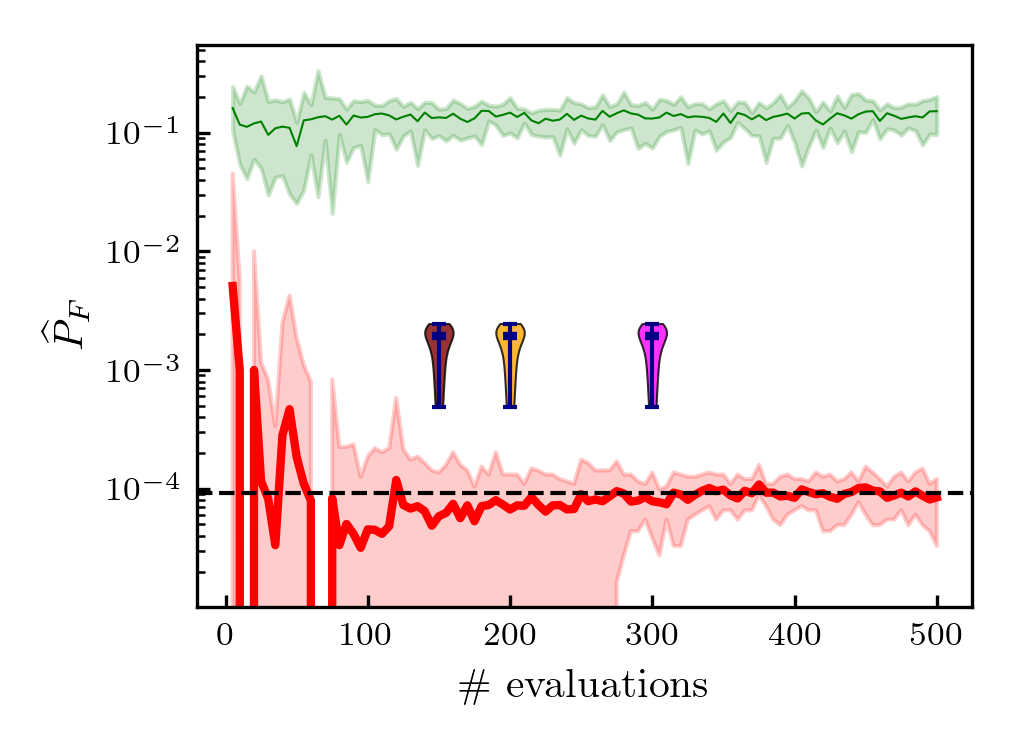}
        \caption{Herbie ($t=2.122$)}
    \end{subfigure}\\
    \begin{subfigure}{.5\textwidth}
        \includegraphics[width=1\linewidth]{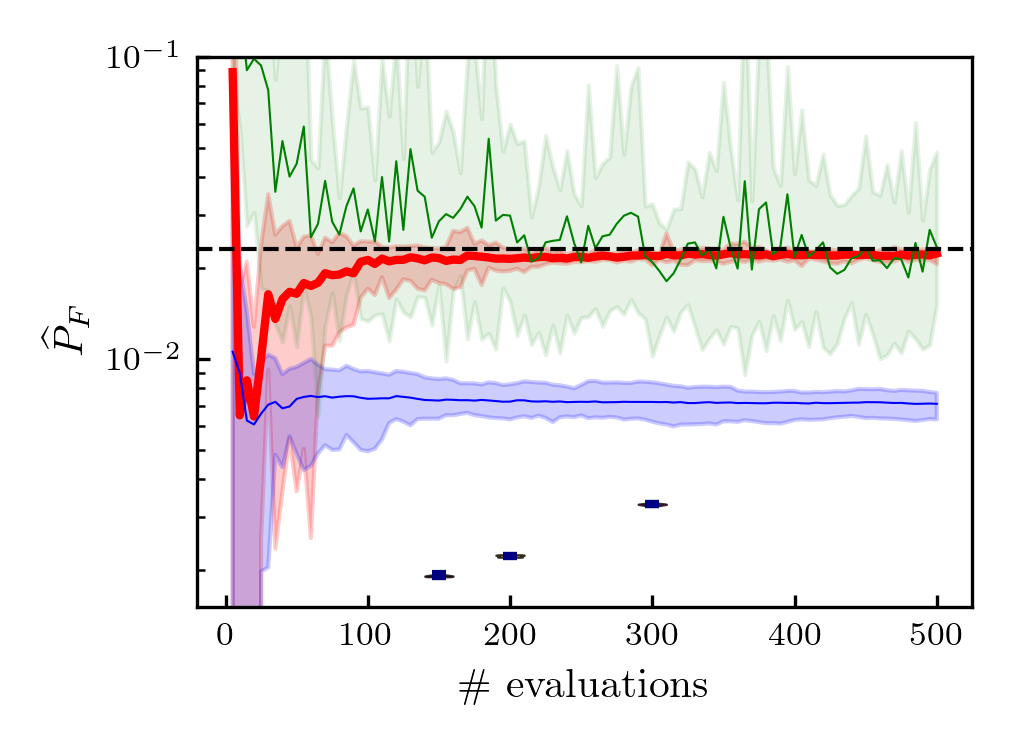}
        \caption{Four branch ($t=2.0$)}
    \end{subfigure}%
    \begin{subfigure}{.5\textwidth}
        \includegraphics[width=1\linewidth]{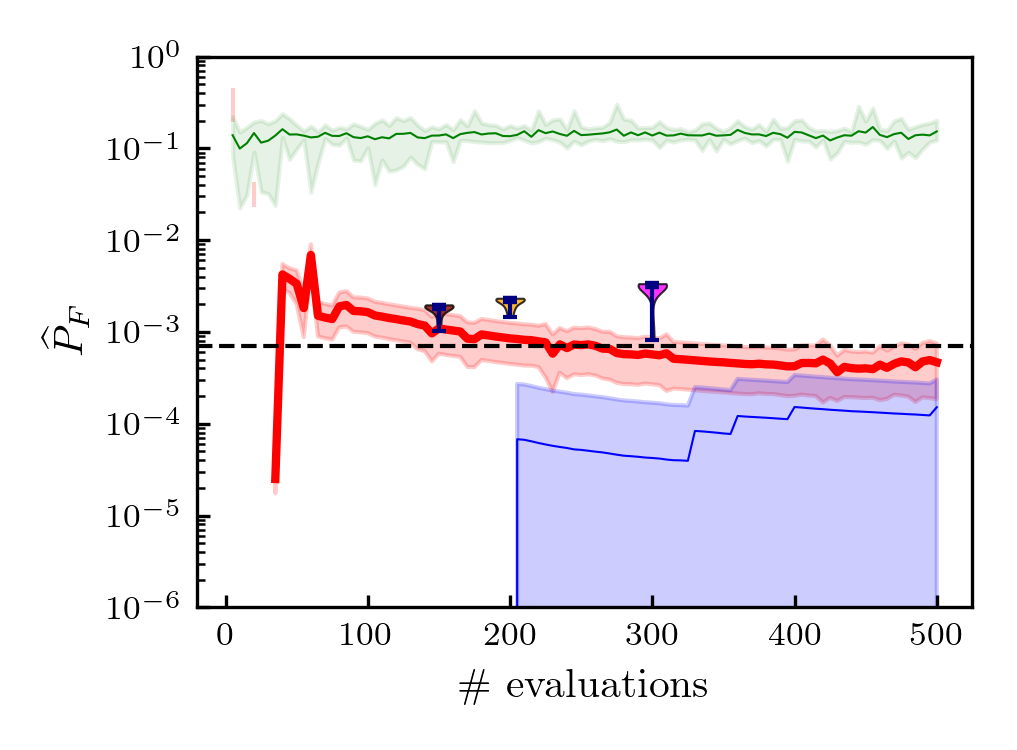}
         \caption{Four branch ($t=3.1$)}
    \end{subfigure}\\
    \begin{subfigure}{1\textwidth}
        \includegraphics[width=1\linewidth]{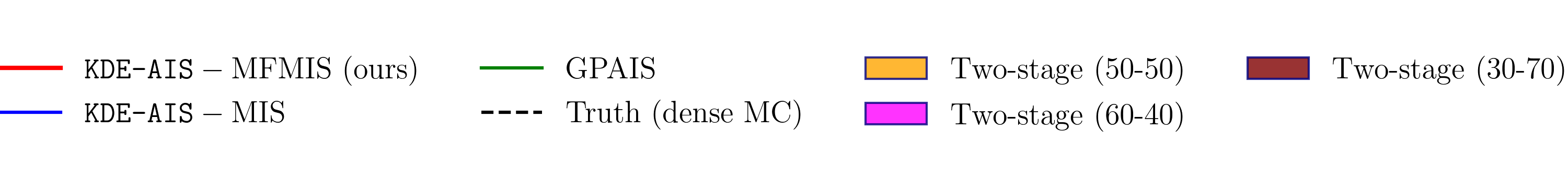}        
    \end{subfigure}
    \caption{Evolution of $\widehat P_F$ against the number of oracle evaluations for the synthetic experiments. Since the two-stage procedures are not sequential, their variability across repetitions is indicated by violin plots at the number of evaluations used for surrogate training. Shaded regions indicate $[\min, \max]$ range.}
    \label{fig:synthetic}
\end{figure}

\subsubsection{Four branch function}
\label{sec:four_branch}
As a second synthetic experiment, we consider the classical four branch
function, also widely used in structural reliability analysis
\cite{schobi2017pck,wang2016gpfs}. For
$\x = (x_1,x_2)^\top \in \mathcal{X} = [-8,8]^2$, the underlying limit-state function
$g : \mathbb{R}^2 \to \mathbb{R}$ is defined as
\[
  g(\x)
  \;=\;
  \min\big\{ g_1(\x), g_2(\x), g_3(\x), g_4(\x) \big\},
\]
with the four branches
\begin{align*}
  g_1(\x)
  &= 3 + 0.1 (x_1 - x_2)^2
     - \frac{x_1 + x_2}{\sqrt{2}},
  \quad
  g_2(\x)
  = 3 + 0.1 (x_1 - x_2)^2
     + \frac{x_1 + x_2}{\sqrt{2}},
  \\
  g_3(\x)
  &= (x_1 - x_2)
     + \frac{7}{\sqrt{2}},
  \quad
  g_4(\x)
  = (x_2 - x_1)
     + \frac{7}{\sqrt{2}}.
\end{align*}

As in the Herbie experiment, we set $p$ to be uniform in $\mcl{X}$ and start the algorithm with $N_0=5$ points, running it for $n=100$ iterations with batches of size $5$. The final comparison of the predicted $g$ and the proposal against the corresponding truths is shown in the right side of \Cref{fig:synthetic_final} -- the conclusions are the same as those made for the Herbie experiment.  \Cref{fig:fb_evolution} shows acquisitions in the same style as \Cref{fig:herbie_evolution}.
The $\widehat P_F$ history is shown in the second row of \Cref{fig:synthetic}, for two different thresholds $t=2$ and $t=3.1$, resulting in (true) failure probabilities $0.0231$ and $0.00071096$, respectively. Similar to the Herbie experiment, the proposed \texttt{KDE-AIS} estimator leads to the most accurate estimate with the least variance, while costing only a few hundred evaluations of the limit state.

\subsection{Real-world experiments}
\begin{figure}[ht!]
    \centering
    \begin{subfigure}{.5\textwidth}
        \includegraphics[width=1\linewidth]{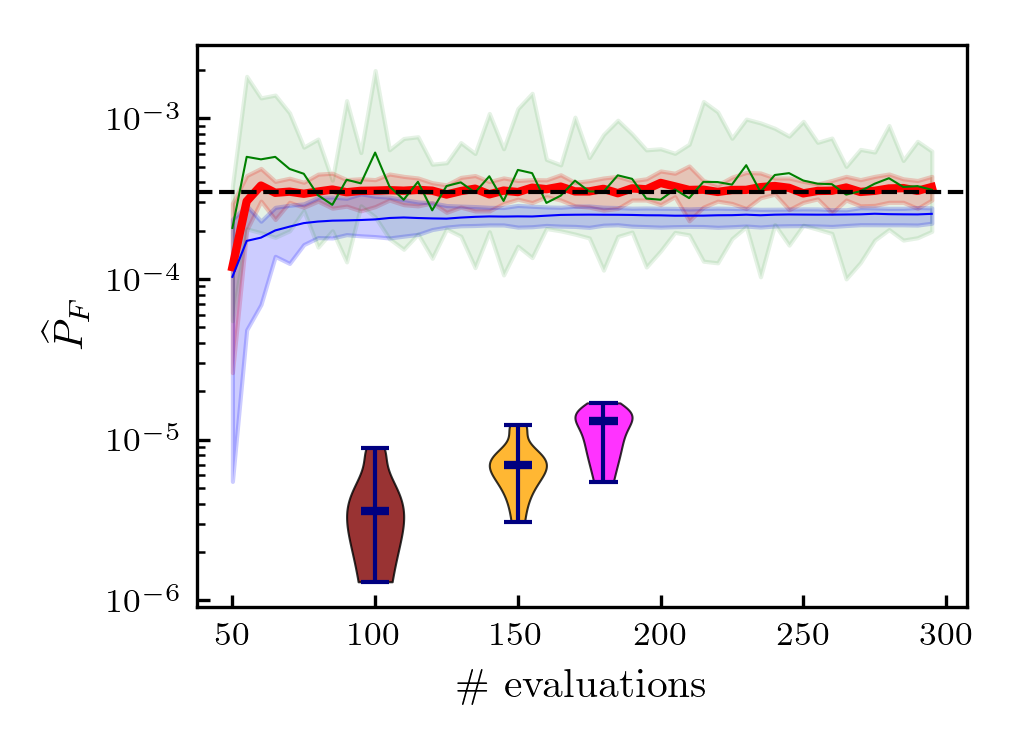}
        \caption{Cantilever beam}
    \end{subfigure}%
    \begin{subfigure}{.5\textwidth}
        \includegraphics[width=1\linewidth]{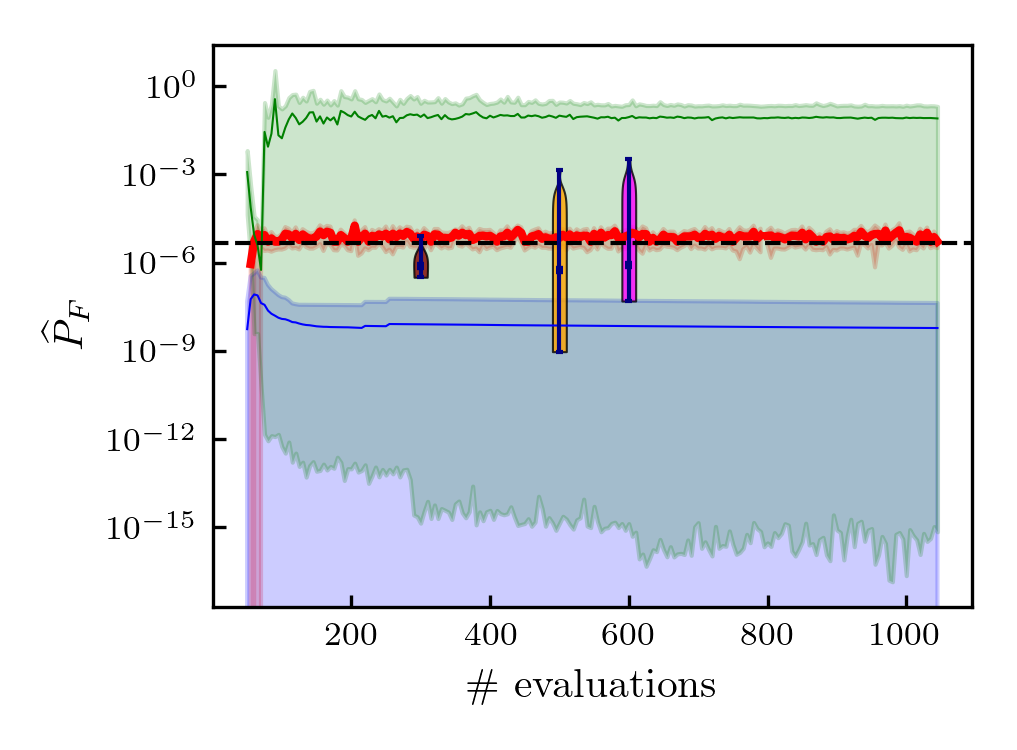}
        \caption{Shaft torsion}
    \end{subfigure}\\
    \begin{subfigure}{1\textwidth}
        \includegraphics[width=1\linewidth]{figures/legend.png}        
    \end{subfigure}
    \caption{Evolution of $\widehat P_F$ against the number of oracle evaluations for the synthetic experiments. Shaded regions indicate $[\min, \max]$ range.}
    \label{fig:realworld}
\end{figure}

\subsubsection{Cantilever beam}
\label{sec:cantilever}
Next, we consider a prismatic cantilever beam under end loads, where the maximum deflection of the beam under the load is used to assess failure. The input vector is
\[
\x=(P,L,E,\Theta)^\top \in\mathbb{R}^4,
\]
where $P$ is the end load, $L$ is the span, $E$ is the Young’s modulus, and $\Theta$ is the thickness. The second moment of area is $I(\Theta)\;=\;\frac{b\,\Theta^3}{12}$,
and the tip deflection under the end load is
\[
\delta(\mathbf{x})\;=\;\frac{P\,L^3}{3\,E\,I(\Theta)} \;=\; \frac{4\,P\,L^3}{E\,b\,\Theta^3}\,.
\]
The limit state function is then defined as
\[
g(\x) = \delta(\x) - D_\mathrm{max},
\]
where $D_\mathrm{max}$ is the maximum displacement, and hence $g(\x) > 0$ is treated as failure. We assume independence and define the input density as
\[
p(\mathbf{x}) \;=\; p_P(p)\,p_L(\ell)\,p_E(e)\,p_{\Theta}(\theta).
\]
A convenient and widely used specification (units in SI) is:
\[
P \sim \mathcal{N}(\mu_P,\sigma_P^2),\qquad
L \sim \mathcal{U}[L_\ell,L_u],\qquad
E \sim \log\mathcal{N}(m_E,s_E^2),\qquad
\Theta \sim \mathcal{U}[\Theta_\ell,\Theta_u],
\]
where
\[
b=0.30\ \text{m},\quad D_{\max}=0.02\ \text{m},\quad
\mu_P=10^4\ \text{N},\ \sigma_P=2\times 10^2\ \text{N},\quad
L\sim\mathcal{U}[3.0,3.1]\ \text{m},
\]
\[
\bar E=2.1\times 10^{11}\ \text{Pa},\ \mathrm{cv}_E=0.05\quad\Rightarrow\quad (m_E,s_E^2)\ \text{as above},\qquad
\Theta\sim\mathcal{U}[0.10,0.20]\ \text{m}.
\]
 A dense MC with $500,000$ samples, repeated independently $100$ times resulted in a $P_F = 0.00035$. The left panel of \Cref{fig:realworld} shows the evolution of $\widehat P_F$, where the proposed $\texttt{KDE-AIS}$ procedure estimates it accurately in about $75$ evaluations.

\subsubsection{Solid round shaft under combined bending and torsion}\label{subsec:shaft_bt}

Finally, we consider a solid circular shaft subjected to combined bending and torsion~\cite{nayek2014reliability}. The input vector is
\[
\x=(M,\;T,\;d,\;\sigma_y,\;G)^\top \in\R^5,
\]
where \(M\) is the bending moment (N\(\cdot\)m), \(T\) is the torque (N\(\cdot\)m), \(d\) is the shaft diameter (m), \(\sigma_y\) is the material yield strength (Pa), and \(G\) is the shear modulus (Pa). The length of the shaft is \(L=1.2\,\text{m}\), the yield safety factor is \(S_F=1.5\), and the twist limit is \(\theta_{\max}=0.06\,\text{rad}\). The failure limit state is
\begin{equation}\label{eq:shaft-limit-state}
g(\x)=\max\!\left(\frac{\sigma_{\T{vm}}(\x)}{\sigma_{\T{allow}}(\x)},\;\frac{\theta(\x)}{\theta_{\max}}\right),
\end{equation}
and $g(\x) > 1$ is considered failure. The stresses and twists are computed relations for a solid circular section given as follows:
\[
\sigma_b(\x) = \frac{32\,M}{\pi\,d^3}, 
\tau(\x) = \frac{16\,T}{\pi\,d^3},
\sigma_{\T{vm}}(\x) = \sqrt{\sigma_b(\x)^2+3\,\tau(\x)^2}, 
\sigma_{\T{allow}}(\x) = \frac{\sigma_y}{S_F},
\theta(\x) = \frac{32\,T\,L}{G\,\pi\,d^4}.
\]
Hence, failure occurs either by yielding \(\big(\sigma_{\T{vm}}>\sigma_{\T{allow}}\big)\) or by excessive twist \(\big(\theta>\theta_{\max}\big)\), whichever is more critical in \eqref{eq:shaft-limit-state}. As in the cantilever beam experiment, 
we take the components of \(\x\) to be independent under \(p\), with
\[
p(\x)=p_M(M)\,p_T(T)\,p_d(d)\,p_{\sigma_y}(\sigma_y)\,p_G(G).
\]
For the loads, we use lognormal models specified as follows:
\[
M\sim\T{LogNormal}(\mu_M,\sigma_M),\quad \mu_M=\ln(450),\;\;\sigma_M=\sqrt{\ln(1+0.25^2)},
\]
\[
T\sim\T{LogNormal}(\mu_T,\sigma_T),\quad \mu_T=\ln(300),\;\;\sigma_T=\sqrt{\ln(1+0.30^2)},
\]
with density \(p_{\T{LN}}(x)=(x\,\sigma)^{-1}(2\pi)^{-1/2}\exp\!\big(-(\ln x-\mu)^2/(2\sigma^2)\big)\) for \(x>0\). For geometry and material, we use truncated normals:
\[
d\sim \T{TN}\!\left(\mu=d_{\T{nom}},\,\sigma=5\!\times\!10^{-4};\; a=d_{\T{nom}}-0.002,\;b=d_{\T{nom}}+0.002\right),
\]
\[
\sigma_y\sim \T{TN}\!\left(\mu=370\!\times\!10^6,\,\sigma=30\!\times\!10^6;\; a=250\!\times\!10^6,\;b=500\!\times\!10^6\right),
\]
\[
G\sim \T{TN}\!\left(\mu=80\!\times\!10^9,\,\sigma=3\!\times\!10^9;\; a=70\!\times\!10^9,\;b=90\!\times\!10^9\right).
\]
Here \(\T{TN}(\mu,\sigma;a,b)\) denotes a Normal\((\mu,\sigma^2)\) truncated to \([a,b]\) with density
\[
p_{\T{TN}}(x)=\frac{\phi\!\left(\frac{x-\mu}{\sigma}\right)}{\sigma\big(\Phi(\frac{b-\mu}{\sigma})-\Phi(\frac{a-\mu}{\sigma})\big)}\;\mathbbm{1}\{a\le x\le b\},
\]

The orders of magnitude difference in the scale of the variables poses a unique challenge in this experiment. A dense MC ($500,000$ samples from $p$) estimate, repeated $100$ times, resulted in a failure probability $P_F = 4.7e-6$. As shown in the right panel of \Cref{fig:realworld}, the proposed $\texttt{KDE-AIS}$ estimator predicts this well with fewer than $100$ evaluations.

\section{Conclusions}
\label{sec:conclusions}

In this work, we have proposed \emph{kernel density estimation adaptive importance sampling} (\texttt{KDE-AIS}), a single-stage sample-efficient framework for estimating rare-event failure probabilities for expensive black-box limit-state functions. Unlike classical two-stage surrogate-assisted approaches that first fit a global surrogate for the limit state and then, in a separate step, construct a biasing density, \texttt{KDE-AIS} treats the design of the importance sampling proposal as the primary goal. The method uses a Gaussian process surrogate for the limit state to construct soft failure probabilities $\pi_n(\x)$, and combines them with the input density $p(\x)$ via a weighted kernel density estimator to approximate the zero-variance proposal $q^\star(\x) \propto p(\x)\,\mathbbm{1}_{\{g(\x) > t\}}$. A slowly vanishing exploration mixture with $p(\x)$ guarantees asymptotically dense sampling of $\mathcal{X}$ and protects against missing isolated or low-probability failure regions. Crucially, the surrogate for the limit state and the proposal for the optimal IS density are learned from the \emph{same} oracle evaluations, which underpins the sample efficiency of our method compared to existing two-stage approaches such as~\cite{peherstorfer2016multifidelity} and Gaussian process based adaptive importance sampling \cite{dalbey2014gpais}.

On the theoretical side, we established that, under mild regularity assumptions on $p$, the surrogate, and the KDE, the KDE-based proposal $q_n$ converges in total variation to the optimal IS density $q^\star$. This is achieved despite the presence of both surrogate error and density-estimation error and is controlled through a suitable exploration schedule and bandwidth choice. We further showed that the multifidelity multiple importance sampling estimator based on the full history of proposals is unbiased for every finite sampling budget and that its variance converges to the oracle variance associated with $q^\star$. These results formally justify the single-stage design and clarify how the GP surrogate, KDE, and exploration mechanism work together to yield an asymptotically optimal importance sampler.

Numerical experiments on synthetic benchmarks (the Herbie and Four Branch functions) and on two engineering reliability problems (a cantilever beam under end loading and a solid round shaft under combined bending and torsion) demonstrate the practical benefits of \texttt{KDE-AIS}. Across these examples, \texttt{KDE-AIS} recovers proposals that are visually and quantitatively close to $q^\star$ with only a few hundred evaluations of $g$ and produces failure probability estimates with substantially reduced variance compared with a single-fidelity MIS, a single-proposal IS, and another GP based approach in the literature, (GPAIS)~\cite{dalbey2014gpais}. 

{\bf Known limitations.} Despite these advantages, \texttt{KDE-AIS} has several limitations that are important to acknowledge. First, the method is built around GP surrogates and KDE, both of which scale poorly with ambient dimension and the number of design points. However, specific approaches exist to scale GPs and KDE to the high-dimensional setting which we plan to explore in the future. Second, the theoretical guarantees rely on smoothness and support assumptions on $p$ and on the failure set; strongly non-smooth limit-state functions, discontinuities, or highly anisotropic behavior may violate these conditions and slow down convergence to $q^\star$. However, we did not face them in the 
experiments investigated in this work. Third, \texttt{KDE-AIS} is designed for settings where $p$ is known and easy to sample from and where the inputs are continuous; discrete, mixed, or strongly constrained design spaces are not handled natively. 

Future work will focus on addressing these limitations and broadening the scope of \texttt{KDE-AIS}. On the modeling side, replacing the KDE with higher-capacity transport-based or normalizing-flow proposals and combining them with sparse or low-rank GP surrogates offers a promising route to improving scalability in moderate to high dimensions while retaining theoretical guarantees. From an algorithmic perspective, it will be important to design adaptive bandwidth and exploration schedules that are tuned online to the evolving surrogate uncertainty and the estimated failure probability, and to develop non-asymptotic performance bounds that explicitly quantify the number of model evaluations required in the rare-event regime. On the application side, integrating \texttt{KDE-AIS} into reliability-based design optimization loops and extending it to system-level and time-dependent reliability problems are natural next steps. 

\appendix
\section{Appendix}

\subsection{Unbiasedness of the MIS balance–heuristic estimator}

\begin{proof}
Index the samples so that $X_i\sim q_{k(i)}$ for a known assignment
$k(i)\in\{0,\dots,n\}$, where $k(i)=0$ denotes the $N_0$ initial draws from $p$.
By linearity of expectation,
\[
\mathbb E\!\left[\widehat P_F^{\,\mathrm{MIS}}\right]
=\frac{1}{N_{\mathrm{tot}}}\sum_{i=1}^{N_{\mathrm{tot}}}
\mathbb E_{q_{k(i)}}\!\left[
\mathbbm{1}_{\{g(X)>t\}}\frac{p(X)}{\bar q_{N_{\mathrm{tot}}}(X)}\right].
\]
Each expectation is an integral under its own proposal:
\[
\mathbb E_{q_{k(i)}}\!\left[
\mathbbm{1}_{\{g(X)>t\}}\frac{p(X)}{\bar q_{N_{\mathrm{tot}}}(X)}\right]
= \int_{\mathcal X}\mathbbm{1}_{\{g(\x)>t\}}\,
\frac{p(\x)}{\bar q_{N_{\mathrm{tot}}}(\x)}\,q_{k(i)}(\x)\,d\x.
\]
Summing over $i$ and dividing by $N_{\mathrm{tot}}$ gives
\[
\mathbb E\!\left[\widehat P_F^{\,\mathrm{MIS}}\right]
= \int_{\mathcal X}\mathbbm{1}_{\{g(\x)>t\}}\,p(\x)\,
\underbrace{\frac{ \frac{1}{N_{\mathrm{tot}}}\sum_{i=1}^{N_{\mathrm{tot}}} q_{k(i)}(\x) }
{ \bar q_{N_{\mathrm{tot}}}(\x) }}_{=\,1}\,d\x,
\]
because by definition
$\frac{1}{N_{\mathrm{tot}}}\sum_{i=1}^{N_{\mathrm{tot}}} q_{k(i)}(\x)
=\frac{n_0}{N_{\mathrm{tot}}}p(\x)+\sum_{k=0}^{n}\frac{N_k}{N_{\mathrm{tot}}}q_k(\x)
= \bar q_{N_{\mathrm{tot}}}(\x)$.
Therefore $\mathbb E[\widehat P_F^{\,\mathrm{MIS}}]=\int \mathbbm{1}_{\{g>t\}}p
= P_F$. 
\end{proof}

\subsection{Proof of \Cref{prop:estimator_for_rn} (unbiased estimation of surrogate error)}
\label{sec:estimator_for_rn}
\begin{proof}
Write \(\pi=\pi_n(\x)\in[0,1]\), \(I=\mathbbm 1_F(\x)\in\{0,1\}\). Since \(I\) is binary and \(\pi\in[0,1]\), we have the pointwise identity
\[
|\pi-I|
=\begin{cases}
1-\pi, & \text{on }F\ (I=1),\\[2pt]
\ \ \ \pi, & \text{on }F^{c}\ (I=0),
\end{cases}
\]
which is equivalently written as
\[
|\pi-I| \;=\; I\,(1-\pi) \;+\; (1-I)\,\pi.
\]
Taking the expectation under \(p\) gives the decomposition
\[
r_n
= \mathbb E_p\!\big[\,\mathbbm 1_F(1-\pi_n)\,\big]
  + \mathbb E_p\!\big[\,\big(1-\mathbbm 1_F\big)\pi_n\,\big]
= \int_F (1-\pi_n)\,p \;+\; \int_{F^c} \pi_n\,p .
\]
Expanding the indicators also yields the equivalent form
\[
r_n
= \mathbb E_p[\mathbbm 1_F] + \mathbb E_p[\pi_n] - 2\,\mathbb E_p[\pi_n\,\mathbbm 1_F]
= P_F + \mathbb E_p[\pi_n] - 2\,\mathbb E_p[\pi_n\,\mathbbm 1_F].
\]

Now, let $\omega_i \;=\; \frac{p(\mathbf{x}_i)}{\bar q_{N_{\text{tot}}}(\mathbf{x}_i)}$ be the importance weights under the MIS estimator and $z_i \;=\; \mathbbm{1}_{\{y_i>t\}}.$ Then, 


\[
\widehat{\mathbb{E}_p\!\big[\mathbbm{1}_F(1-\pi_n)\big]}
\;=\;
\frac{1}{N_{\text{tot}}}\sum_{i=1}^{N_{\text{tot}}}
z_i\,\big(1-\pi_n(\mathbf{x}_i)\big)\,\omega_i,
\qquad
\widehat{\mathbb{E}_p\!\big[(1-\mathbbm{1}_F)\pi_n\big]}
\;=\;
\frac{1}{N_{\text{tot}}}\sum_{i=1}^{N_{\text{tot}}}
(1-z_i)\,\pi_n(\mathbf{x}_i)\,\omega_i,
\] leads to
\[
\widehat r_n
\;=\;
\widehat{\mathbb{E}_p\!\big[\mathbbm{1}_F(1-\pi_n)\big]}
\;+\;
\widehat{\mathbb{E}_p\!\big[(1-\mathbbm{1}_F)\pi_n\big]}.]
\]
which can be further decomposed with
\[
\widehat P_F
\;=\;
\frac{1}{N_{\text{tot}}}\sum_{i=1}^{N_{\text{tot}}}
z_i\,\omega_i,
\qquad
\widehat{\mathbb{E}_p[\pi_n]}
\;=\;
\frac{1}{N_{\text{tot}}}\sum_{i=1}^{N_{\text{tot}}}
\pi_n(\mathbf{x}_i)\,\omega_i,
\qquad
\widehat{\mathbb{E}_p[\pi_n\,\mathbbm{1}_F]}
\;=\;
\frac{1}{N_{\text{tot}}}\sum_{i=1}^{N_{\text{tot}}}
\pi_n(\mathbf{x}_i)\,z_i\,\omega_i,
\]
to give
\[
\boxed{\;
\widehat r_n
\;=\;
\widehat P_F
\;+\;
\widehat{\mathbb{E}_p[\pi_n]}
\;-\;
2\,\widehat{\mathbb{E}_p[\pi_n\,\mathbbm{1}_F]}
\;}
\]
\end{proof}

\subsection{Proof of \Cref{lem:mfmis-var-comparison}}
\label{sec:mfmis_var_proof}

\begin{proof}
Define
\[
Z_i
:=
\mathbbm{1}_{\{\widehat g_n(\mathbf{x}_i)>t\}}
\frac{p(\mathbf{x}_i)}{\bar q_{N_{\mathrm{tot}}}(\mathbf{x}_i)},
\qquad
Y_i
:=
\left[
\mathbbm{1}_{\{ g_n(\mathbf{x}_i)>t\}}
-
\mathbbm{1}_{\{\widehat g_n(\mathbf{x}_i)>t\}}
\right]
\frac{p(\mathbf{x}_i)}{\bar q_{N_{\mathrm{tot}}}(\mathbf{x}_i)}.
\]
Let
\[
\widetilde S_n
:=
\frac{1}{N_{\mathrm{tot}}}\sum_{i=1}^{N_{\mathrm{tot}}} Z_i,
\]
so that the regular MIS estimator can be written as
\[
\widehat P_{F,n}^{\mathrm{MIS}}=\widetilde S_n+R_n.
\]
Likewise, if \(S_n\) is formed from \(M_{\mathrm{tot}}\) independent cheap surrogate samples, then
\[
\widehat P_{F,n}^{\mathrm{MF\text{-}MIS}}=S_n+R_n,
\]
with \(S_n\) independent of \(R_n\).

By construction,
\[
\operatorname{Var}(\widetilde S_n)=\frac{V_{S,n}}{N_{\mathrm{tot}}},
\qquad
\operatorname{Var}(S_n)=\frac{V_{S,n}}{M_{\mathrm{tot}}},
\qquad
\operatorname{Cov}(\widetilde S_n,R_n)=\frac{C_n}{N_{\mathrm{tot}}}.
\]
Therefore,
\[
\operatorname{Var}\!\left(\widehat P_{F,n}^{\mathrm{MIS}}\right)
=
\operatorname{Var}(\widetilde S_n)
+
\operatorname{Var}(R_n)
+
2\,\operatorname{Cov}(\widetilde S_n,R_n)
=
\frac{V_{S,n}}{N_{\mathrm{tot}}}
+
\operatorname{Var}(R_n)
+
\frac{2}{N_{\mathrm{tot}}}C_n,
\]
whereas
\[
\operatorname{Var}\!\left(\widehat P_{F,n}^{\mathrm{MF\text{-}MIS}}\right)
=
\operatorname{Var}(S_n)
+
\operatorname{Var}(R_n)
=
\frac{V_{S,n}}{M_{\mathrm{tot}}}
+
\operatorname{Var}(R_n).
\]
Subtracting gives
\[
\operatorname{Var}\!\left(\widehat P_{F,n}^{\mathrm{MIS}} \right)
-
\operatorname{Var}\!\left(\widehat P_{F,n}^{\mathrm{MF\text{-}MIS}} \right)
=
\left(\frac{1}{N_{\mathrm{tot}}}-\frac{1}{M_{\mathrm{tot}}}\right)V_{S,n}
+\frac{2}{N_{\mathrm{tot}}}C_n.
\]
Hence
\[
\operatorname{Var}\!\left(\widehat P_{F,n}^{\mathrm{MF\text{-}MIS}} \right)
\le
\operatorname{Var}\!\left(\widehat P_{F,n}^{\mathrm{MIS}} \right)
\]
whenever
\[
\left(\frac{1}{N_{\mathrm{tot}}}-\frac{1}{M_{\mathrm{tot}}}\right)V_{S,n}
+\frac{2}{N_{\mathrm{tot}}}C_n \ge 0,
\]
that is,
\[
C_n \ge -\frac12\left(1-\frac{N_{\mathrm{tot}}}{M_{\mathrm{tot}}}\right)V_{S,n}.
\]
\end{proof}
\subsection{Proof of Theorem 1}
\label{sec:thm1_proof}
\begin{proof}[Proof of Theorem~\ref{thm:prop-conv}]
Write $F=\{\x:g(\x)>t\}$.
Let $Z_n:=\int_{\mathcal X}\pi_n(\x)^\alpha p(\x)\,d\x$ and $P_F:=\int_{\mathcal X}\mathbbm 1_F(\x)\,p(\x)\,d\x$.
Define the \emph{unnormalized} measures
\[
\mu_n(d\x):=\pi_n(\x)^\alpha p(\x)\,d\x,
\qquad
\mu^\star(d\x):=\mathbbm 1_F(\x)\,p(\x)\,d\x,
\]
so that the \emph{normalised} densities are
\[
q_n^\dagger=\frac{d\mu_n}{Z_n},\qquad q^\star=\frac{d\mu^\star}{P_F}.
\]
We proceed in three steps.

\medskip
\noindent\textbf{Step 1 (Plug-in convergence $q_n^\dagger\Rightarrow q^\star$).}
Since $u\mapsto u^\alpha$ is $\alpha$-H\"older on $[0,1]$ with constant $1$ -- that is, $|u^\alpha - v^\alpha| \leq |u - v|^\alpha,~\forall u, v \in [0,1]$ -- we state that  for all $\x$
\[
\big|\pi_n(\x)^\alpha-\mathbbm 1_F(\x)^\alpha \big| = \big|\pi_n(\x)^\alpha-\mathbbm 1_F(\x)\big|
\;\le\;
\big|\pi_n(\x)-\mathbbm 1_F(\x)\big|^\alpha.
\]
Integrating against $p$, and using the total variation distance identity between probability measures, yields
\[
\|\mu_n-\mu^\star\|_{\mathrm{TV}}
=\tfrac12\!\int |\pi_n^\alpha-\mathbbm 1_F|\,p(\x) d\x
\;\le\;\tfrac12\,\|\pi_n-\mathbbm 1_F\|_{L^1(p)}^\alpha
=\tfrac12\,r_n^\alpha.
\]
Next, we want to bound the total variation distance between
the normalized densities
\[
q_n^\dagger(\mathbf{x}) 
:= \frac{\pi_n(\mathbf{x})^\alpha p(\mathbf{x})}{Z_n},
\qquad
q^\star(\mathbf{x}) 
:= \frac{\mathbbm{1}_F(\mathbf{x})\,p(\mathbf{x})}{P_F}.
\]

The total variation (TV) distance between the probability measures
with densities $q_n^\dagger$ and $q^\star$ is
\begin{equation*}
\begin{split}
\| q_n^\dagger - q^\star \|_{\mathrm{TV}}
=& \frac{1}{2} \int_{\mcl{X}}
\bigl| q_n^\dagger(\mathbf{x}) - q^\star(\mathbf{x}) \bigr|\,d\mathbf{x} \\
=& \frac{1}{2} \int_{\mcl{X}}
\left| \frac{\pi_n(\x)^\alpha p(\x)}{Z_n} - \frac{\mathbbm{1}_F(\x) p(\x)}{P_F}\right|\,d\mathbf{x} \\
=& \frac{1}{2} \int_{\mcl{X}}
\left| \frac{|\pi_n(\x)^\alpha - \mathbbm{1}_F(\x) | p(\x)}{Z_n P_F} - \frac{\mathbbm{1}_F(\x) p(\x) (P_F - Z_n)}{Z_n P_F}\right|\,d\mathbf{x} \\
\leq & \frac{1}{2} \int_{\mcl{X}}
\left| \frac{|\pi_n(\x)^\alpha - \mathbbm{1}_F(\x) | p(\x)}{Z_n P_F}\right| - \left| \frac{|P_F - Z_n| \mathbbm{1}_F(\x) p(\x) }{Z_n P_F}\right|\,d\mathbf{x} \\
=& \frac{1}{2Z_n P_F} \int_\mcl{X} |\pi_n(\x)^\alpha - \mathbbm{1}_F(\x)| p(\x) d\x + \frac{1}{2 Z_n P_F} |P_F - Z_n| \int_\mcl{X} \mathbbm{1}_F(\x) p(\x) d\x \\
=& \frac{1}{2Z_n P_F} \int_\mcl{X} |\pi_n(\x)^\alpha - \mathbbm{1}_F(\x)| p(\x) d\x + \frac{1}{2 Z_n} |P_F - Z_n| \\
\leq & C_1 r_n^\alpha + C_2 r_n^\alpha.
\end{split}
\end{equation*}
The fourth line in the previous step is due to the triangle inequality, and we have used the fact that $|Z_n-P_F| = \int |\pi_n(\x)^\alpha - \mathbbm{1}_F(\x)|p  \le \int |\pi_n-\mathbbm 1_F|^\alpha\,p = r_n^\alpha$.

\paragraph{Step 2 (KDE convergence $\widehat{q}_n \rightarrow q_n^\dagger$).} Our next step is to show that the KDE converges to the surrogate proposal $q_n^\dagger$. From the pilot samples $\{\mathbf{u}_j\}_{j=1}^{m_n}$, i.i.d.\ $\sim p$,
and the bandwidth $h_n \downarrow 0$, we define the weighted KDE
\[
\widehat{q}_n(\mathbf{x})
\;:=\;
\sum_{j=1}^{m_n} \tilde{w}_{n,j}\,\varphi_{h_n}(\mathbf{x}-\mathbf{u}_j),
\qquad
\tilde{w}_{n,j} \propto \bigl(\pi_n(\mathbf{u}_j)\bigr)^\alpha,
\]
where $\varphi_{h_n}(\mathbf{x}) = h_n^{-d}\,\varphi(\mathbf{x}/h_n)$ and
$\varphi:\mathbb{R}^d\to\mathbb{R}$ is a bounded Lipschitz kernel integrating to $1$.
Assumption~4 further states that the (normalized) weights satisfy $0 \;\le\; \tilde{w}_{n,j} \;\le\; 1, \qquad
\sum_{j=1}^{m_n} \tilde{w}_{n,j} = 1,$
and that under these conditions, the weighted KDE inherits the uniform consistency
rates of the standard KDE. Our goal in this step is to prove that
\[
\bigl\|\widehat{q}_n - q_n^\dagger\bigr\|_\infty
\;=\;
O_p\!\left(
  \sqrt{\frac{\log(m_n)}{m_n h_n^{d}}}
  \;+\;
  h_n^\beta
\right),
\]
where $\beta>0$ is the Hölder exponent from Assumption~1.

For notational convenience, let us write $f_n(\mathbf{x}) := q_n^\dagger(\mathbf{x})$
for the (normalized) surrogate target density; our ultimate objective is to bound
$\bigl\|\widehat{q}_n - q^\dagger_n\bigr\|_\infty$.
We can write the ideal kernel-smoothed target as
\[
(q^\dagger_n * \varphi_{h_n})(\mathbf{x})
:= \int_{\mcl{X}} \varphi_{h_n}(\mathbf{x}-\mathbf{y})\,q^\dagger_n(\mathbf{y})\,d\mathbf{y}.
\]
Note that, for samples drawn from the true density $q_n^\dagger$, the expectation of the KDE at $\x$ is given by
\[\mbb{E}\big[ q_n^\dagger(\x) \big] = \int \varphi(\x- \mbf{y}) q_n^\dagger(\mbf{y}) d \mbf{y} = (q_n^\dagger * \varphi_{h_n})(\x),\]
and hence we call $(q_n^\dagger * \varphi_{h_n})(\x)$ the ``ideal'' target.

We then add and subtract this smoothed target inside the difference:
\begin{align*}
\widehat{q}_n(\mathbf{x}) - q^\dagger_n(\mathbf{x})
&=
\underbrace{
  \Bigl[\widehat{q}_n(\mathbf{x}) - (q^\dagger_n * \varphi_{h_n})(\mathbf{x})\Bigr]
}_{\text{stochastic / variance term}}
+
\underbrace{
  \Bigl[(q^\dagger_n * \varphi_{h_n})(\mathbf{x}) - q^\dagger_n(\mathbf{x})\Bigr]
}_{\text{deterministic bias term}}.
\end{align*}

Taking the supremum over $\mathbf{x}\in\mcl{X}$ and using the triangle inequality, we obtain
\begin{equation}
\label{eq:step2-main-decomp}
\bigl\|\widehat{q}_n - q^\dagger_n\bigr\|_\infty
\;\le\;
\bigl\|\widehat{q}_n - q^\dagger_n * \varphi_{h_n}\bigr\|_\infty
+
\bigl\|q^\dagger_n * \varphi_{h_n} - q^\dagger_n\bigr\|_\infty.
\end{equation}

Thus, we must bound each of the two terms on the right-hand side.

\medskip

\textbf{Step 2.2: Bias term $\bigl\|q^\dagger_n * \varphi_{h_n} - q^\dagger_n\bigr\|_\infty$.}

We now show that the bias term is of order $O(h_n^\beta)$ under the Hölder regularity
assumption.
Assumption~1 states that $p$ is $\beta$-Hölder on $\mcl{X}$. For this step we assume
(in line with the informal reasoning in the theorem) that $q_n^\dagger$ is also
$\beta$-Hölder on $\mcl{X}$, that is, there exists $L<\infty$ (independent of $n$) such that
\[
|q^\dagger_n(\mathbf{x}) - q^\dagger_n(\mathbf{y})|
\;\le\;
L\,\|\mathbf{x}-\mathbf{y}\|^\beta,
\qquad
\forall\,\mathbf{x},\mathbf{y}\in\mcl{X}.
\]

Fix $\mathbf{x}\in\mcl{X}$. Write
\begin{align*}
(q^\dagger_n * \varphi_{h_n})(\mathbf{x}) - q^\dagger_n(\mathbf{x})
&=
\int_{\mcl{X}} \varphi_{h_n}(\mathbf{x}-\mathbf{y})
\bigl[q^\dagger_n(\mathbf{y}) - q^\dagger_n(\mathbf{x})\bigr]\,d\mathbf{y}.
\end{align*}
Taking absolute values gives
\begin{align*}
\bigl|(f_n * \varphi_{h_n})(\mathbf{x}) - f_n(\mathbf{x})\bigr|
&\le
\int_{\mcl{X}} \varphi_{h_n}(\mathbf{x}-\mathbf{y})
\bigl|f_n(\mathbf{y}) - f_n(\mathbf{x})\bigr|\,d\mathbf{y}.
\end{align*}
Using the $\beta$-Hölder continuity of $f_n$,
\[
\bigl|q^\dagger_n(\mathbf{y}) - q^\dagger_n(\mathbf{x})\bigr|
\;\le\;
L\,\|\mathbf{y}-\mathbf{x}\|^\beta,
\]
we obtain
\begin{align*}
\bigl|(q^\dagger_n * \varphi_{h_n})(\mathbf{x}) - q^\dagger_n(\mathbf{x})\bigr|
&\le
L \int_{\mcl{X}} \varphi_{h_n}(\mathbf{x}-\mathbf{y})
\|\mathbf{y}-\mathbf{x}\|^\beta\,d\mathbf{y}.
\end{align*}
Now perform the change of variables
\[
\mathbf{z} = \frac{\mathbf{x}-\mathbf{y}}{h_n},
\qquad
\mathbf{y} = \mathbf{x} - h_n \mathbf{z},
\qquad
d\mathbf{y} = h_n^d\,d\mathbf{z}.
\]
Since $\varphi_{h_n}(\mathbf{x}-\mathbf{y}) = h_n^{-d}\,\varphi(\mathbf{z})$, we have
\begin{align*}
\int_{\mcl{X}} \varphi_{h_n}(\mathbf{x}-\mathbf{y})
\|\mathbf{y}-\mathbf{x}\|^\beta\,d\mathbf{y}
&=
\int_{\mathbb{R}^d} h_n^{-d}\,\varphi(\mathbf{z})
\bigl\|h_n\mathbf{z}\bigr\|^\beta\,h_n^d\,d\mathbf{z} \\
&=
h_n^\beta
\int_{\mathbb{R}^d} \varphi(\mathbf{z})\,\|\mathbf{z}\|^\beta\,d\mathbf{z}.
\end{align*}
By assumption $\varphi$ is bounded and integrable, and $\|\mathbf{z}\|^\beta$ grows at most
polynomially, so the integral
\[
C_\varphi := \int_{\mathbb{R}^d} \varphi(\mathbf{z})\,\|\mathbf{z}\|^\beta\,d\mathbf{z}
\]
is finite. Therefore
\[
\bigl|(q^\dagger_n * \varphi_{h_n})(\mathbf{x}) - q^\dagger_n(\mathbf{x})\bigr|
\;\le\;
L C_\varphi\,h_n^\beta,
\qquad
\forall\,\mathbf{x}\in\mcl{X}.
\]
Taking the supremum over $\mathbf{x}\in\mcl{X}$ yields
\[
\bigl\|q^\dagger_n * \varphi_{h_n} - q^\dagger_n\bigr\|_\infty
\;\le\;
L C_\varphi\,h_n^\beta
=
O(h_n^\beta).
\]

Thus the bias term in \eqref{eq:step2-main-decomp} is of order $O(h_n^\beta)$.

\medskip

\paragraph{Step 2.3: Stochastic term $\bigl\|\widehat{q}_n - q^\dagger_n * \varphi_{h_n}\bigr\|_\infty$.} using standard results in KDE~\cite{gine2004weighted,Tsybakov2009} with our assumptions (that $\varphi$ is bounded and Lipschitz), one can show that the \emph{unnormalized} KDE $\widehat q_n(\x) = \frac{1}{m_n} \sum_{j=1}^{m_n} \varphi_{h_n}(\x - \mbf{y}_j)$ error is bounded by 
\[
\bigl|\widehat{q}_n - q^\dagger_n * \varphi_{h_n}\bigr| \leq \bigl\|\widehat{q}_n - q^\dagger_n * \varphi_{h_n}\bigr\|_\infty
=
O_p\!\left(
  \sqrt{\frac{\log(m_n)}{m_n h_n^{d}}}
\right).
\]





\smallskip

\emph{Weighted KDE case.} Invoking Assumption 4, which asserts that the weighted KDE enjoys the same convergence rate as the unweighted case, we have 
\[
\bigl\|\widehat{q}_n - q^\dagger_n * \varphi_{h_n}\bigr\|_\infty
\;\le\;
C\,
\sqrt{\frac{\log(m_n)}{m_n h_n^{d}}}
\qquad\text{in probability as }n\to\infty.
\]




\textbf{Step 2.4: Combine bias and stochastic terms.}

Returning to the decomposition
\eqref{eq:step2-main-decomp}, we now plug in the bounds obtained in Steps~2.2 and~2.3:
\begin{align*}
\bigl\|\widehat{q}_n - q_n^\dagger\bigr\|_\infty
&=
\bigl\|\widehat{q}_n - q^\dagger_n\bigr\|_\infty \\
&\le
\bigl\|\widehat{q}_n - q^\dagger_n * \varphi_{h_n}\bigr\|_\infty
+
\bigl\|q^\dagger_n * \varphi_{h_n} - q^\dagger_n\bigr\|_\infty \\
&=
O_p\!\left(
  \sqrt{\frac{\log(m_n)}{m_n h_n^{d}}}
\right)
+
O(h_n^\beta).
\end{align*}
Hence
\[
\bigl\|\widehat{q}_n - q_n^\dagger\bigr\|_\infty
=
O_p\!\left(
  \sqrt{\frac{\log(m_n)}{m_n h_n^{d}}}
  \;+\;
  h_n^\beta
\right),
\]
which is exactly the claimed KDE convergence rate in Step~2 of Theorem~1.

\medskip
\noindent\textbf{Step 3 (Mixture closeness and conclusion).}
By definition,
\[
q_n = (1-\eta_n)\,\widehat q_n + \eta_n\,p,
\qquad
q_n - q_n^\dagger
= (1-\eta_n)(\widehat q_n - q_n^\dagger) + \eta_n(p - q_n^\dagger).
\]
Hence, using $\|\cdot\|_{L^1}\le \lambda(\mathcal X)\,\|\cdot\|_\infty$ on a compact domain, where $\lambda(\mathcal X)$ is the Lebesgue measure of the domain $\mathcal{X}$, and
$\|p-q_n^\dagger\|_{L^1}\le 2$ since $p$ and $q_n^\dagger$ are densities,
\[
\|q_n - q_n^\dagger\|_{\mathrm{TV}}
\;\le\; \tfrac12(1-\eta_n)\,\lambda(\mathcal X)\,\|\widehat q_n - q_n^\dagger\|_\infty
\;+\; \eta_n
\;\xrightarrow[n\to\infty]\;0,
\]
since $\eta_n\to0$ and the KDE term vanishes by Step~2.
Finally, by the triangle inequality,
\[
\|q_n - q^\star\|_{\mathrm{TV}}
\;\le\; \|q_n - q_n^\dagger\|_{\mathrm{TV}} + \|q_n^\dagger - q^\star\|_{\mathrm{TV}}
\;\xrightarrow[n\to\infty]{}\;0,
\]
which proves all three displayed claims. Since our surrogate estimate $q_n$ converges to $q^\star$, necessarily $\widehat P_{F, n}^\T{MIS}$ and $\widehat P_{F, n}^\T{MF-MIS}$ converge to $P_F$ with vanishing variance asymptotically.
\end{proof}

\clearpage
\bibliographystyle{plainnat}
\bibliography{sample, ref, other_refs, cv, gp_reliability, kde_is_fp}

@article{echard2011akmcs,
  title={{AK-MCS}: An active learning reliability method combining Kriging and Monte Carlo simulation},
  author={Echard, Benjamin and Gayton, Nicolas and Lemaire, Michel},
  journal={Structural Safety},
  volume={33},
  number={2},
  pages={145--154},
  year={2011},
  doi={10.1016/j.strusafe.2011.01.002}
}

@article{dubourg2013mbis,
  title={Metamodel-based importance sampling for structural reliability analysis},
  author={Dubourg, Vincent and Deheeger, Franck and Sudret, Bruno},
  journal={Probabilistic Engineering Mechanics},
  volume={33},
  pages={47--57},
  year={2013},
  doi={10.1016/j.probengmech.2013.02.002}
}

@article{balesdent2013akis,
  title={Kriging-based adaptive importance sampling algorithms for rare event estimation},
  author={Balesdent, Mathieu and Morio, Jérôme and Marzat, Julien},
  journal={Structural Safety},
  volume={44},
  pages={1--10},
  year={2013},
  doi={10.1016/j.strusafe.2013.05.002}
}

@article{cadini2014improvedAKIS,
  title={An improved adaptive Kriging-based importance technique for sampling multiple failure regions of low probability},
  author={Cadini, Francesco and Santos, Agnese and Zio, Enrico},
  journal={Reliability Engineering \& System Safety},
  volume={131},
  pages={109--117},
  year={2014},
  doi={10.1016/j.ress.2014.06.023}
}

@article{huang2016akss,
  title={Assessing small failure probabilities by AK--SS: An active learning method combining Kriging and Subset Simulation},
  author={Huang, Xianfeng and Chen, Jie and Li, Su},
  journal={Structural Safety},
  volume={59},
  pages={86--95},
  year={2016},
  doi={10.1016/j.strusafe.2015.12.003}
}

@article{zhang2019akss,
  title={An active learning reliability method combining Kriging and subset simulation},
  author={Zhang, Jing and Tong, Zhen and Li, Yiqun and Zhang, Jun},
  journal={Reliability Engineering \& System Safety},
  volume={188},
  pages={90--102},
  year={2019},
  doi={10.1016/j.ress.2019.03.001}
}

@article{bect2017bayesSubSim,
  title={Bayesian subset simulation},
  author={Bect, Julien and Li, Ling and Vazquez, Emmanuel},
  journal={SIAM/ASA Journal on Uncertainty Quantification},
  volume={5},
  number={1},
  pages={762--786},
  year={2017},
  doi={10.1137/16M1078276}
}

@article{xiao2020aksis,
  title={Reliability analysis with stratified importance sampling based on adaptive Kriging},
  author={Xiao, Sinan and Oladyshkin, Sergey and Nowak, Wolfgang},
  journal={Reliability Engineering \& System Safety},
  volume={197},
  pages={106852},
  year={2020},
  doi={10.1016/j.ress.2019.106852}
}

@article{au2001ss,
  title={Estimation of small failure probabilities in high dimensions by subset simulation},
  author={Au, Siu-Kui and Beck, James L.},
  journal={Probabilistic Engineering Mechanics},
  volume={16},
  number={4},
  pages={263--277},
  year={2001},
  doi={10.1016/S0266-8920(01)00019-4}
}

@article{ang1992kernel,
  title={Optimal importance-sampling density estimator},
  author={Ang, A. H-S. and Tang, W. H. (see also Ang et al.)},
  journal={Journal of Engineering Mechanics},
  volume={118},
  number={6},
  pages={1146--1164},
  year={1992},
  doi={10.1061/(ASCE)0733-9399(1992)118:6(1146)}
}

@article{au1999aiskernel,
  title={A new adaptive importance sampling scheme for reliability evaluation},
  author={Au, Siu-Kui and Beck, James L.},
  journal={Structural Safety},
  volume={21},
  number={2},
  pages={135--158},
  year={1999},
  doi={10.1016/S0167-4730(99)00014-4}
}

@article{lee2017kdeis,
  title={An adaptive importance sampling method with a Kriging model and kernel sampling density for structural reliability analysis},
  author={Lee, Soon Man and Song, Jae-Hun and Kim, Nam-Ho},
  journal={Journal of Mechanical Science and Technology},
  volume={31},
  pages={5873--5882},
  year={2017},
  doi={10.1007/s12206-017-1119-8}
}

@article{li2021npis,
  title={Nonparametric importance sampling for wind turbine reliability analysis},
  author={Li, Shiyu and Peng, Yi and Byon, Eui-Young},
  journal={Annals of Applied Statistics},
  volume={15},
  number={4},
  pages={1850--1873},
  year={2021},
  doi={10.1214/20-AOAS1438}
}

@article{persoons2023akis,
  title={Variance reduction using multiple importance sampling with adaptive Kriging (AK-AMIS)},
  author={Persoons, Augustin and Broggi, Matteo and Beer, Michael},
  journal={PhD Thesis / preprint, University of Liverpool},
  year={2023},
  note={Available as: \url{https://livrepository.liverpool.ac.uk/3171699/1/APersoons.pdf}}
}

@article{dasgupta2024rein,
  title={{REIN}: Reliability estimation via importance sampling with normalizing flows},
  author={Dasgupta, Agnimitra and Johnson, Erik},
  journal={Reliability Engineering \& System Safety},
  volume={242},
  pages={109729},
  year={2024},
  doi={10.1016/j.ress.2023.109729}
}

@article{Breunig2008,
  author  = {Robert Breunig},
  title   = {Nonparametric density estimation for stratified samples},
  journal = {Statistics \& Probability Letters},
  year    = {2008},
  volume  = {78},
  number  = {14},
  pages   = {2194--2200},
  doi     = {10.1016/j.spl.2008.01.099},
}

@article{BuskirkLohr2005,
  author  = {Trent D. Buskirk and Sharon L. Lohr},
  title   = {Asymptotic properties of kernel density estimation with complex survey data},
  journal = {Journal of Statistical Planning and Inference},
  year    = {2005},
  volume  = {128},
  number  = {1},
  pages   = {165--190},
  doi     = {10.1016/j.jspi.2003.09.036},
}

@article{booth2025two,
  title={Two-stage design for failure probability estimation with Gaussian process surrogates},
  author={Booth, Annie S and Renganathan, S Ashwin},
  journal={Journal of Quality Technology},
  pages={1--17},
  year={2025},
  publisher={Taylor \& Francis}
}

@article{der2000geometry,
  title={The geometry of random vibrations and solutions by FORM and SORM},
  author={Der Kiureghian, A},
  journal={Probabilistic Engineering Mechanics},
  volume={15},
  number={1},
  pages={81--90},
  year={2000},
  publisher={Elsevier}
}

@article{hu2021second,
  title={Second-order reliability methods: a review and comparative study},
  author={Hu, Zhangli and Mansour, Rami and Olsson, M{\aa}rten and Du, Xiaoping},
  journal={Structural and multidisciplinary optimization},
  volume={64},
  number={6},
  pages={3233--3263},
  year={2021},
  publisher={Springer}
}

@incollection{reynolds2015gaussian,
  title={Gaussian mixture models},
  author={Reynolds, Douglas},
  booktitle={Encyclopedia of biometrics},
  pages={827--832},
  year={2015},
  publisher={Springer}
}

@article{zhao1999response,
  title={Response uncertainty and time-variant reliability analysis for hysteretic MDF structures},
  author={Zhao, YG and Ono, T and Idota, H},
  journal={Earthquake engineering \& structural dynamics},
  volume={28},
  number={10},
  pages={1187--1213},
  year={1999},
  publisher={Wiley Online Library}
}

@book{madsen2006methods,
  title={Methods of structural safety},
  author={Madsen, Henrik O and Krenk, Steen and Lind, Niels Christian},
  year={2006},
  publisher={Courier Corporation}
}

@article{huang2018reliability,
  title={Reliability analysis of slope stability under seismic condition during a given exposure time},
  author={Huang, HW and Wen, SC and Zhang, J and Chen, FY and Martin, JR and Wang, H},
  journal={Landslides},
  volume={15},
  number={11},
  pages={2303--2313},
  year={2018},
  publisher={Springer}
}

@article{huang2017overview,
  title={Overview of structural reliability analysis methods—Part I: Local reliability methods},
  author={Huang, Changwu and El Hami, Abdelkhalak and Radi, Boucha{\"\i}b},
  journal={Incertitudes et fiabilit{\'e} des syst{\`e}mes multiphysiques},
  volume={17},
  number={1},
  pages={1--10},
  year={2017}
}

@article{chevalier2014fast,
  title={Fast parallel kriging-based stepwise uncertainty reduction with application to the identification of an excursion set},
  author={Chevalier, Cl{\'e}ment and Bect, Julien and Ginsbourger, David and Vazquez, Emmanuel and Picheny, Victor and Richet, Yann},
  journal={Technometrics},
  volume={56},
  number={4},
  pages={455--465},
  year={2014},
  publisher={Taylor \& Francis}
}

@article{azzimonti2021adaptive,
  title={Adaptive design of experiments for conservative estimation of excursion sets},
  author={Azzimonti, Dario and Ginsbourger, David and Chevalier, Cl{\'e}ment and Bect, Julien and Richet, Yann},
  journal={Technometrics},
  volume={63},
  number={1},
  pages={13--26},
  year={2021},
  publisher={Taylor \& Francis}
}

@article{duhamel2023version,
  title={A SUR version of the Bichon criterion for excursion set estimation},
  author={Duhamel, Cl{\'e}ment and Helbert, C{\'e}line and Munoz Zuniga, Miguel and Prieur, Cl{\'e}mentine and Sinoquet, Delphine},
  journal={Statistics and Computing},
  volume={33},
  number={2},
  pages={41},
  year={2023},
  publisher={Springer}
}

@article{gine2004weighted,
  title={Weighted uniform consistency of kernel density estimators},
  author={Gine, Evarist and Koltchinskii, Vladimir and Zinn, Joel},
  journal={The Annals of Probability},
  volume={32},
  number={3B},
  pages={2570--2605},
  year={2004},
  doi={10.1214/009117904000000063}
}

@article{lee2011herbie,
  author  = {Lee, H. K. H. and Gramacy, R. B. and Linkletter, C. and Gray, G. A.},
  title   = {Optimization Subject to Hidden Constraints via Statistical Emulation},
  journal = {Pacific Journal of Optimization},
  volume  = {7},
  number  = {3},
  pages   = {467--478},
  year    = {2011}
}

@article{dalbey2014gpais,
  author  = {Dalbey, Keith R. and Swiler, Laura P.},
  title   = {Gaussian Process Adaptive Importance Sampling},
  journal = {International Journal for Uncertainty Quantification},
  volume  = {4},
  number  = {2},
  year    = {2014},
  doi     = {10.1615/Int.J.UncertaintyQuantification.2013006330}
}

@inproceedings{romero2016pof,
  author       = {Romero, Vicente J. and Swiler, Laura P. and Ebeida, Mohamed S. and Mitchell, Scott A.},
  title        = {A Set of Test Problems and Results in Assessing Method Performance for Calculating Low Probabilities of Failure},
  booktitle    = {18th AIAA Non-Deterministic Approaches Conference},
  year         = {2016},
  address      = {San Diego, CA},
  organization = {AIAA},
  note         = {AIAA Paper 2016-0429},
  doi          = {10.2514/6.2016-0429}
}

@article{schobi2017pck,
  author  = {Sch\"obi, R. and Sudret, B. and Marelli, S.},
  title   = {Rare Event Estimation Using Polynomial-Chaos Kriging},
  journal = {ASCE-ASME Journal of Risk and Uncertainty in Engineering Systems, Part A: Civil Engineering},
  volume  = {3},
  number  = {2},
  year    = {2017},
  doi     = {10.1061/AJRUA6.0000870}
}

@article{wang2016gpfs,
  author  = {Wang, Hongqiao and Lin, Guang and Li, Jinglai},
  title   = {Gaussian Process Surrogates for Failure Detection: A {B}ayesian Experimental Design Approach},
  journal = {Journal of Computational Physics},
  volume  = {313},
  pages   = {247--259},
  year    = {2016},
  doi     = {10.1016/j.jcp.2016.02.053}
}

@article{Rosenblatt1956, author={Rosenblatt, M.}, title={Remarks on some nonparametric estimates of a density function}, journal={Ann. Math. Statist.}, year={1956}, volume={27}, number={3}, pages={832--837}}

@article{Parzen1962, author={Parzen, E.}, title={On estimation of a probability density function and mode}, journal={Ann. Math. Statist.}, year={1962}, volume={33}, number={3}, pages={1065--1076}}

@book{silverman1986, author={Silverman, B. W.}, title={Density Estimation for Statistics and Data Analysis}, publisher={Chapman \& Hall}, year={1986}}

@book{Scott2015, author={Scott, D. W.}, title={Multivariate Density Estimation}, edition={2}, publisher={Wiley}, year={2015}}

@book{WandJones1995, author={Wand, M. P. and Jones, M. C.}, title={Kernel Smoothing}, publisher={Chapman \& Hall}, year={1995}}

@article{TerrellScott1992, author={Terrell, G. R. and Scott, D. W.}, title={Variable kernel density estimation}, journal={Ann. Statist.}, year={1992}, volume={20}, number={3}, pages={1236--1265}}

@book{Tsybakov2009, author={Tsybakov, A. B.}, title={Introduction to Nonparametric Estimation}, publisher={Springer}, year={2009}}

@article{Botev2010, author={Botev, Z. I. and Grotowski, J. F. and Kroese, D. P.}, title={Kernel density estimation via diffusion}, journal={Ann. Statist.}, year={2010}, volume={38}, number={5}, pages={2916--2957}}

@book{Owen2013, author={Owen, A. B.}, title={Monte Carlo Theory, Methods and Examples}, publisher={Stanford University}, year={2013}}

@inproceedings{VeachGuibas1995, author={Veach, E. and Guibas, L. J.}, title={Optimally combining sampling techniques for Monte Carlo rendering}, booktitle={Proc. SIGGRAPH}, year={1995}, pages={419--428}}

@article{ElviraMIS, author={Elvira, V. and Martino, L. and Luengo, D. and Bugallo, M. F.}, title={Generalized multiple importance sampling}, journal={arXiv:1511.03095}, year={2019}}

@article{jones1998efficient,
  doi          = {10.1023/A:1008306431147},
  title={Efficient global optimization of expensive black-box functions},
  author={Jones, Donald R. and Schonlau, Matthias and Welch, William J.},
  journal={Journal of Global Optimization},
  volume={13},
  number={4},
  pages={455--492},
  year={1998},
  publisher={Springer}
}

@book{rasmussen:williams:2006,
  doi       = {10.7551/mitpress/3206.001.0001},
  author = {Rasmussen, C. E. and Williams, C. K. I.},
  publisher = {MIT Press},
  title = {Gaussian Processes for Machine Learning},
  year = 2006
}

@book{gramacy2020surrogates,
  title = {Surrogates: {G}aussian Process Modeling, Design and Optimization for the Applied Sciences},
  author = {Robert B. Gramacy},
  publisher = {Chapman Hall/CRC},
  address = {Boca Raton, Florida},
  url = {http://bobby.gramacy.com/surrogates},
  year = {2020}
}

@book{santner2003design,
	doi = {10.1007/978-1-4757-3799-8},
	year = 2003,
	author = {Thomas J. Santner and Brian J. Williams and William I. Notz},
	title = {The Design and Analysis of Computer Experiments},
  publisher={Springer}
}

@article{nayek2014reliability,
  title={Reliability approximation for solid shaft under Gamma setup},
  author={Nayek, Sadananda and Seal, Babulal and Roy, Dilip},
  journal={Journal of Reliability and Statistical Studies},
  pages={11--17},
  year={2014}
}

@incollection{silverman1986kernel,
  title={The kernel method for univariate data},
  author={Silverman, Bernard W},
  booktitle={Density estimation for statistics and data analysis},
  pages={34--74},
  year={1986},
  publisher={Springer}
}

@article{epanechnikov1969non,
  title={Non-parametric estimation of a multivariate probability density},
  author={Epanechnikov, Vassiliy A},
  journal={Theory of Probability \& Its Applications},
  volume={14},
  number={1},
  pages={153--158},
  year={1969},
  publisher={SIAM}
}

@article{booth2025contour,
  title={Contour location for reliability in airfoil simulation experiments using deep gaussian processes},
  author={Booth, Annie S and Renganathan, S Ashwin and Gramacy, Robert B},
  journal={The Annals of Applied Statistics},
  volume={19},
  number={1},
  pages={191--211},
  year={2025},
  publisher={Institute of Mathematical Statistics}
}

@article{naveau2020statistical,
  title={Statistical methods for extreme event attribution in climate science},
  author={Naveau, Philippe and Hannart, Alexis and Ribes, Aur{\'e}lien},
  journal={Annual Review of Statistics and Its Application},
  volume={7},
  number={1},
  pages={89--110},
  year={2020},
  publisher={Annual Reviews}
}

@article{love2012credible,
  title={Credible occurrence probabilities for extreme geophysical events: Earthquakes, volcanic eruptions, magnetic storms},
  author={Love, Jeffrey J},
  journal={Geophysical Research Letters},
  volume={39},
  number={10},
  year={2012},
  publisher={Wiley Online Library}
}

@article{bugallo2017adaptive,
  title={Adaptive importance sampling: The past, the present, and the future},
  author={Bugallo, Monica F and Elvira, Victor and Martino, Luca and Luengo, David and Miguez, Joaquin and Djuric, Petar M},
  journal={IEEE Signal Processing Magazine},
  volume={34},
  number={4},
  pages={60--79},
  year={2017},
  publisher={IEEE}
}

@article{oh1992adaptive,
  title={Adaptive importance sampling in Monte Carlo integration},
  author={Oh, Man-Suk and Berger, James O},
  journal={Journal of statistical computation and simulation},
  volume={41},
  number={3-4},
  pages={143--168},
  year={1992},
  publisher={Taylor \& Francis}
}

@book{srinivasan2002importance,
  title={Importance sampling: Applications in communications and detection},
  author={Srinivasan, Rajan},
  year={2002},
  publisher={Springer Science \& Business Media}
}

@inproceedings{gotovos2013active,
  title={Active learning for level set estimation},
  author={Gotovos, Alkis and Casati, Nathalie and Hitz, Gregory and Krause, Andreas},
  booktitle={Twenty-Third International Joint Conference on Artificial Intelligence},
  year={2013}
}

@article{li2011efficient,
  title={An efficient surrogate-based method for computing rare failure probability},
  author={Li, Jing and Li, Jinglai and Xiu, Dongbin},
  journal={Journal of Computational Physics},
  volume={230},
  number={24},
  pages={8683--8697},
  year={2011},
  publisher={Elsevier}
}

@article{marques2018contour,
  title={Contour location via entropy reduction leveraging multiple information sources},
  author={Marques, Alexandre N and Lam, Remi R and Willcox, Karen E},
  journal={arXiv preprint arXiv:1805.07489},
  year={2018}
}

@article{cole2021entropy,
  title={Entropy-based adaptive design for contour finding and estimating reliability},
  author={Cole, D Austin and Gramacy, Robert B and Warner, James E and Bomarito, Geoffrey F and Leser, Patrick E and Leser, William P},
  journal={arXiv preprint arXiv:2105.11357},
  year={2021}
}

@article{bichon2008efficient,
  title={Efficient global reliability analysis for nonlinear implicit performance functions},
  author={Bichon, Barron J and Eldred, Michael S and Swiler, Laura Painton and Mahadevan, Sandaran and McFarland, John M},
  journal={AIAA journal},
  volume={46},
  number={10},
  pages={2459--2468},
  year={2008}
}

@article{bect2012sequential,
  title={Sequential design of computer experiments for the estimation of a probability of failure},
  author={Bect, Julien and Ginsbourger, David and Li, Ling and Picheny, Victor and Vazquez, Emmanuel},
  journal={Statistics and Computing},
  volume={22},
  number={3},
  pages={773--793},
  year={2012},
  publisher={Springer}
}

@article{picheny2010adaptive,
  title={Adaptive designs of experiments for accurate approximation of a target region},
  author={Picheny, Victor and Ginsbourger, David and Roustant, Olivier and Haftka, Raphael T and Kim, Nam-Ho},
  year={2010}
}

@article{peherstorfer2016multifidelity,
  title={Multifidelity importance sampling},
  author={Peherstorfer, Benjamin and Cui, Tiangang and Marzouk, Youssef and Willcox, Karen},
  journal={Computer Methods in Applied Mechanics and Engineering},
  volume={300},
  pages={490--509},
  year={2016},
  publisher={Elsevier}
}

@article{ranjan2008sequential,
  title={Sequential experiment design for contour estimation from complex computer codes},
  author={Ranjan, Pritam and Bingham, Derek and Michailidis, George},
  journal={Technometrics},
  volume={50},
  number={4},
  pages={527--541},
  year={2008},
  publisher={Taylor \& Francis}
}

@article{renganathan2023camera,
  title={CAMERA: A method for cost-aware, adaptive, multifidelity, efficient reliability analysis},
  author={Renganathan, S Ashwin and Rao, Vishwas and Navon, Ionel M},
  journal={Journal of Computational Physics},
  volume={472},
  pages={111698},
  year={2023},
  publisher={Elsevier}
}

@inproceedings{booth2024actively,
  title={Actively learning deep Gaussian process models for failure contour and probability estimation.},
  author={Booth, Annie S and Gramacy, Robert and Renganathan, Ashwin},
  booktitle={AIAA SCITECH 2024 Forum},
  pages={0577},
  year={2024}
}

\end{document}